\documentclass[reprint,twocolumn,amsmath,amssymb,superscriptaddress]{revtex4-2}

\usepackage{graphicx, graphpap}
\usepackage[dvipsnames]{xcolor}
\usepackage[caption=false,subrefformat=parens,labelformat=parens]{subfig}
\usepackage{amsmath}
\usepackage{overpic}
\usepackage{physics}
\usepackage{enumitem}
\usepackage{comment}
\usepackage{amssymb}
\usepackage{amsthm}
\usepackage{pstricks}
\usepackage{float}
\usepackage{multirow}
\usepackage{hyperref}
\usepackage[T1]{fontenc}
\usepackage{framed} 
\usepackage{bm}
\usepackage{bbm}
\definecolor{shadecolor}{gray}{0.9}
\definecolor{ECCgreen}{RGB}{0, 153, 0}
\definecolor{EditPurple}{RGB}{160, 32, 240}

\renewcommand{\Tr}[1]{\mathrm{Tr}\left[#1\right]}

\usepackage{placeins}

\newtheorem{theorem}{Theorem}

\newtheorem{lemma}[theorem]{Lemma}

\begin{document}
	
\title{Composable Continuous-Variable Multi-User QKD with Discrete Modulation:\\ Theory and Implementation}
\author{Florian Kanitschar}
\email{florian.kanitschar@outlook.com}
\affiliation{AIT  Austrian  Institute  of  Technology,  Center  for  Digital  Safety\&Security,  Giefinggasse  4,  1210  Vienna, Austria}
\affiliation{Vienna Center for Quantum Science and Technology (VCQ), Atominstitut, Technische Universität Wien, Stadionallee 2, 1020 Vienna, Austria}

\author{Adnan A.E. Hajomer}
\email{aaeha@dtu.dk}
\affiliation{Center for Macroscopic Quantum States (bigQ), Department of Physics, Technical University of Denmark, 2800 Kongens Lyngby, Denmark}

\author{Michael Hentschel}
\affiliation{AIT  Austrian  Institute  of  Technology,  Center  for  Digital  Safety\&Security,  Giefinggasse  4,  1210  Vienna, Austria}

\author{Tobias Gehring}
\affiliation{Center for Macroscopic Quantum States (bigQ), Department of Physics, Technical University of Denmark, 2800 Kongens Lyngby, Denmark}

\author{Christoph Pacher}
\affiliation{AIT  Austrian  Institute  of  Technology,  Center  for  Digital  Safety\&Security,  Giefinggasse  4,  1210  Vienna, Austria}
\affiliation{fragmentiX Storage Solutions GmbH, Wohllebengasse 10, 1040 Vienna, Austria}
\date{\today}
	
\begin{abstract}
Establishing scalable, secure quantum networks requires advancing beyond conventional point-to-point quantum key distribution (QKD) protocols toward point-to-multipoint QKD protocols. Here, we generalize a well-established discrete-modulated continuous-variable (CV) QKD protocol from the point-to-point to the point-to-multipoint setting. We present a comprehensive security analysis across four trust scenarios and derive secret key rates for both loss-only and noisy channels, in the asymptotic and composable finite-size regimes. Experimentally, we validate the protocol in a passive optical network with $10$ km access links, achieving a composable secure key rate of $2.185 \times 10^{-3}$ bits per symbol ($0.273$ Mbit/s) against independent and identically distributed collective attacks. Our results demonstrate that discrete-modulated CV-QKD can support access networks with multiple users while relying solely on cost-efficient, off-the-shelf telecommunication components, paving the way toward practical, scalable, and secure quantum networks.  
\end{abstract}
	
\maketitle



\section{Introduction \label{sec:Intro}} 
Quantum Key Distribution (QKD) \cite{Bennett_Brassard_1984, Ekert_1991} enables two remote parties to establish cryptographic keys with information-theoretic security. Among the various approaches, continuous-variable (CV) QKD protocols \cite{Ralph_1999, Usenko_2025} are particularly attractive because they can be implemented using standard telecommunication technology, such as quadrature modulation and coherent detection, at room temperature. This makes CV-QKD a promising candidate for widespread deployment in metropolitan networks. Nevertheless, while QKD allows the establishment of secure keys between two parties, enabling secure point-to-point connections, modern telecommunication networks interconnect multiple entities. Extending QKD to the multi-user regime is therefore essential for its large-scale deployment and integration into real-world network infrastructures. 

CV-QKD protocols can be broadly classified by their modulation scheme. Gaussian-modulated (GM) protocols \cite{Cerf_2001, Grosshans_2002, Grangier_2002, Silberhorn_2002, Navascues_2006, Garcia_2006, Diamanti_2015} benefit from powerful symmetry arguments that simplify their security analysis \cite{Leverrier_2013, Leverrier_2017}, but Gaussian modulation itself is an idealization that has never been exactly realized in practice. By contrast, discrete-modulated (DM) protocols \cite{Heid_2006, Zhao_2009, Sych_2010} use finite constellations that are well-suited to experimental implementation, though at the cost of more complex security proofs \cite{Ghorai_2019, Lin_2019, Lin_2020, Upadhyaya_2021, Denys_2021, Matsuura_2021, Kanitschar_2021, Lupo_2022, Kanitschar_2023, Baeuml_2023}.


While multi-user QKD has been extensively investigated in the discrete-variable domain \cite{Cabello_2000, Chen_2008, Epping_2017, Grasselli_2018, Murta_2020}, and more recently also with Gaussian-modulated CV-QKD \cite{Bian_2023, Derkach_2024}, the case of discrete-modulated CV-QKD in multi-user networks has remained largely unexplored.


Here, we extend a general DM-CV-QKD protocol, previously studied in the standard point-to-point setting \cite{Ghorai_2019, Lin_2019, Lin_2020, Upadhyaya_2021, Denys_2021, Kanitschar_2021, Lupo_2022, Kanitschar_2023, Baeuml_2023}, to the multi-user regime. The scheme relies on a simple “cheap source” architecture that differs from a conventional single-user source only by the addition of beam splitters, making it compatible with cost-efficient, off-the-shelf components. We adapt an analytical security proof technique \cite{Heid_2006} to evaluate different trust models in loss-only channels, and further extend numerical security proof methods \cite{Coles_2016, Winick_2018, Lin_2019, Kanitschar_2023} to analyze general noisy channels. Finally, we demonstrate the practicality of our protocol in a passive optical network with $10$ km access links, achieving a composable secure key rate of $2.185 \times 10^{-3}$ bits per symbol ($0.273$ Mbit/s) against independent and identically distributed collective attacks.

\section{Theory}\label{sec:Theory}
As soon as QKD is performed in a network, multiple users are involved. While in point-to-multipoint networks, a single Alice prepares and distributes coherent states  (Fig.~\ref{fig:Sketch_Trust_Scenarios}(a)), the roles of the $N_B$ Bobs in the network can differ significantly. For the general DM CV-QKD protocol considered can be found in Figure \ref{fig:Sketch_Trust_Scenarios} (b), we may group the different trust assumptions between the communicating parties into four different trust scenarios

\begin{itemize}
    \item[a)] Alice and one particular Bob trust each other but distrust all other $N_B-1$ Bobs. Alice aims to establish secret keys with the trusted Bob. 
    \item[b)] Alice and a group of $1 < M_B < N_B$ Bobs trust each other but distrust all other $N_B-M_B$ Bobs. Alice aims to establish secret keys with one of the trusted Bobs.
    \item[c)] Alice and all $N_B$ Bobs trust each other. Alice aims to establish secret keys with one (trusted) Bob.
    \item[d)] Alice trusts a group of $1\leq M_B \leq N_B$ Bobs. The $M_B$ Bobs are assumed to be legitimate parties, i.e., they do not collaborate with Eve, but also aim for their own private key with Alice, independently of all other Bobs. Alice aims to simultaneously establish secret keys with all trusted Bobs. 
\end{itemize}
Note that a) and c) are special cases of b) for $M_B = 1$ and $M_B = N_B$, respectively. Without loss of generality, we can assume that the first $M_B$ Bobs are trusted while the remaining $N_B-M_B$ are untrusted. The classical protocol phase, which encompasses the last two steps of the protocol (Fig. \ref{fig:Sketch_Trust_Scenarios} (b)),  in trust-scenario d) fundamentally differs from that in cases a) - c), as Alice establishes different secret keys with all mutually trusting Bobs, i.e., during error-correction, each of those Bobs sends syndromes $S_i$ over the classical channel. However, the transmitted syndrome of Bob $i$ can leak information about the raw key of Bob $j\neq i$, which must be considered in the calculation of the key rate. One way to avoid this leakage is to encrypt the syndromes sent via the classical channel \cite{Bian_2023}. 

\begin{figure*}[t]
\centering
\includegraphics[width=2\columnwidth]{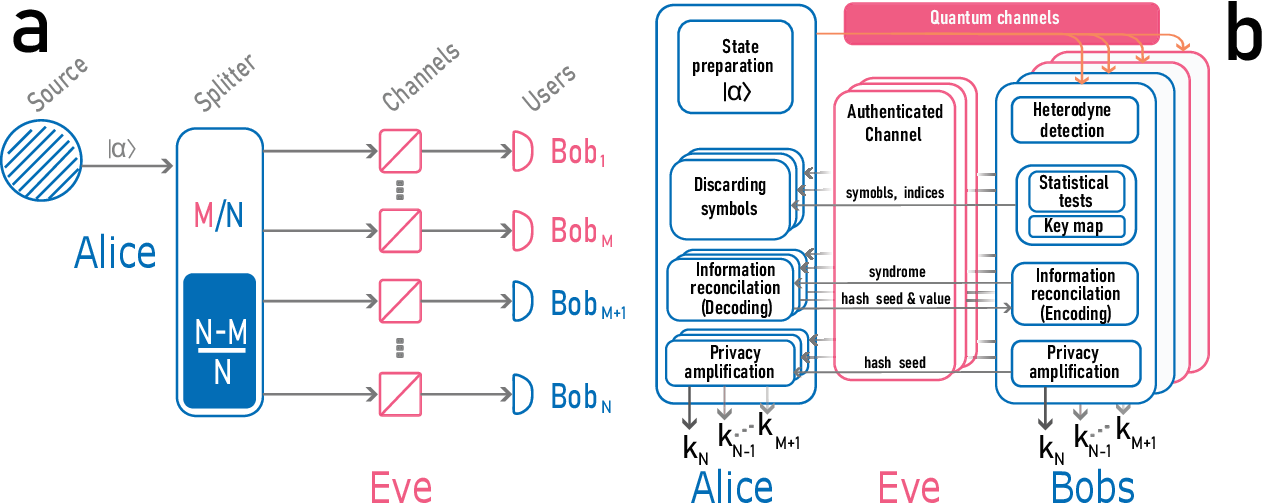}
\caption{(a) Illustration of different trust scenarios, with a) and c) as special cases of b) for $M_B = 1$ and $M_B = N_B$. (b) Protocol description. \label{fig:Sketch_Trust_Scenarios}}
\end{figure*}

To account for different trust assumptions within a quantum network, Alice and the Bobs follow slightly adapted versions of the protocol described in Figure \ref{fig:Sketch_Trust_Scenarios}.
Yet, the secure key rates achievable under these varying scenarios remain largely unexplored. To address this gap, we extend established security proof techniques from the well-studied point-to-point case to the more general multi-user setting. 

Denoting the Bob with whom Alice aims to establish a secret key with $B_j$, the asymptotic key rate for cases a) - c) is given by the well-known Devetak-Winter formula~\cite{Devetak_Winter_2006}
\begin{equation}\label{eq:DevetakWinterSimple}
    R^{\infty} = \beta I(A:B_j) - \chi(B_j:E),
\end{equation}
where $\beta$ is the reconciliation efficiency, $I(A:B_j)$ denotes the mutual information between Alice and Bob $j$, and $\chi(B_j:E)$ is the Holevo information, a quantity bounding Eve's information about the trusted Bob's quantum state. 

In case d), where each of the Bobs aims to generate a separate secret key with Alice, the information leaked by the classical communication of other Bobs about Bob $j$'s key needs to be taken into account. As discussed in Ref.~\cite{Bian_2023}, in this case, the Holevo information needs to be replaced by the maximum of the mutual information held by Eve or any of the other Bobs, leading to the modified key rate bound
\begin{equation}\label{eq:DevetakWinterD}
    R^{\infty} = \beta I(A:B_j) - \max\left\{ \chi(B_j:E), \max_{i\ne j}\{I(B_j:B_i)\}  \right\}.
\end{equation}

The main task remaining is determining the Holevo bounds. 
In what follows, we first present an elegant way of calculating those quantities for loss-only channels, followed by a treatment of the general case. 

\subsection{Lossy, noise-free channel}
For the loss-only case, we generalized the proof method from Ref. \cite{Heid_2006} to the multi-user case. In the absence of noise, Eve's optimal attack is known to be the generalized beamsplitting attack: If Alice prepares a coherent state $\ket{\alpha_x}$ and transmits it over a channel with loss parameter $\eta$, the Bobs receives the state $\ket{\sqrt{\eta}\alpha_x}$, while Eve's state is $\ket{\epsilon_x}:=\ket{\sqrt{1-\eta}\alpha_x}$, where $x \in \left[N_{\mathrm{St}}\right]_{-1}$ labels the state Alice prepared. Our first approach is shown in detail in Appendix \ref{sec:DirectAnaProof}, and expresses Eve's conditional state as a superposition of suitably chosen basis states. Thanks to an orthogonality relation that we proved in Lemma \ref{lem:Lemma1}, this reduces the evaluation of the Holevo quantity to the calculation of the coefficients of Eve's conditional state in this basis. However, the computational complexity of the calculation of these coefficients is $\mathcal{O}\left(N_{\mathrm{St}}^{N_B -1} \right)$, i.e., it scales exponentially in the number of Bobs, which makes the evaluation of the Holevo quantity time-consuming already for medium double-digit values of $N_B$. 

Therefore, in a second step, detailed in Appendix \ref{sec:ReductionSingleBob}, we refined our approach further and showed how unitary transformations which respect the symmetry of the setup and leave the Holevo quantity invariant can be used to transform the multi-Bob scenario into the single-Bob scenario, which allows us to efficiently evaluate the Holevo information and, hence, the secure key rate.

\subsection{Lossy and noisy channel}
While the loss-only scenario offers valuable insight, serves as a benchmark, and allows a quick evaluation of multi-user scenarios, it represents an idealized setting. Thus, we extended our analysis to the general case, encompassing both loss and noise. We analyze the general case for arbitrary discrete modulation within the numerical security proof framework of Refs.~\cite{Coles_2016,Winick_2018}, which have been applied to DM CV QKD in Refs.~\cite{Lin_2019, Upadhyaya_2021, Kanitschar_2023} and generalize these ideas to the multi-user case. In the presence of noise, Eve's optimal eavesdropping strategy is unknown; thus, we require fundamentally different methods. The core idea is to phrase the calculation of the involved entropic quantities as a semi-definite program and pose the experimental observations as constraints. However, as the mathematical description of DM CV-QKD protocols involves infinite-dimensional Hilbert spaces, while we aim to evaluate and solve the semi-definite program on a computer with finite-dimensional arithmetic, we apply the dimension reduction method \cite{Upadhyaya_2021} to find an equivalent finite-dimensional representation. In the asymptotic setting, this leads to the following expression for the secure key rate
\begin{equation}
    R^{\infty} = \min_{\rho \in \mathcal{S}} H\left(Z|E\right)_{\Phi(\rho)} - \delta^{\mathrm{EC}} - \Delta(W),
\end{equation}
where $Z$ is the final key, $\mathcal{S}$ includes the constraints, $\delta^{\mathrm{EC}}$ is the leakage due to error-correction, and $\Delta(W)$ is a correction term due to the reduction to finite dimensions. The protocol description and information about the multi-user scenario enter the map $\Phi$, which, along with other details of the security proof, is detailed in Appendix \ref{sec:NumPrfMethod}. Similarly to the developments in the two-user scenario ~\cite{Lin_2020}, we may take imperfect detection with trusted detectors into account, where the trusted detection efficiency is denoted by $\eta_D$ and the trusted electronic noise by $\nu_{\mathrm{el}}$.

Having characterized the protocol’s asymptotic security, it remains to investigate composable security in the finite-size regime. We quantify the security of a real QKD protocol by its trace-distance to an ideal instance of the same protocol. Obeying this notion of security, as well as considering the effects of finite statistics, requires modifications to both the asymptotic security argument and the protocol. Based on observations, the key-generating Bobs carry out an energy test. The associated energy testing theorem allows us to make a statistical statement about the weight of the involved quantum states outside of the finite-dimensional cutoff space we introduced for computational purposes, which adds an extra step to the asymptotic protocol. Along with an acceptance test, which puts the statistical analysis of the experimental observations on a solid footing, and the multi-user generalizations developed for the asymptotic case, an accordingly generalized version of the security proof in Ref. \cite{Kanitschar_2023}, one can prove composable security of a general DM CV QKD multi-user protocol against collective i.i.d. attacks with key rate
\begin{equation}\label{eq:KR_formula_th}
\begin{aligned}
    \frac{\ell}{N} \leq \frac{n}{N} &\left[ \min_{\rho \in \mathcal{S}^{\mathrm{E\&A}}} H(X|E')_{\Phi(\rho)} - \delta(\bar{\epsilon}) - \Delta(W) \right]     \\
    & - \delta_{\mathrm{leak}}^{\mathrm{EC}} - \frac{2}{N} \log_2\left( \frac{1}{\epsilon_{\mathrm{PA}}} \right),
\end{aligned}
\end{equation}
where $\delta^{\mathrm{EC}}_{\mathrm{leak}}$ takes the classical error correction cost into account, $\Delta(W)$ is given in Eq.~(\ref{eq:WeightCorrectionTerm}), $\delta(\epsilon) := 2 \log_2\left( \mathrm{rank}(\rho_X)+3 \right) \sqrt{\frac{\log_2\left(2/\epsilon \right)}{n}}$ and $\mathcal{S}^{\mathrm{E\&A}}$ is defined below, with security parameter $\epsilon_{\mathrm{EC}} + \max\left\{\frac{1}{2}\epsilon_{\mathrm{PA}}+\bar{\epsilon}, \epsilon_{\mathrm{ET}}+\epsilon_{\mathrm{AT}} \right\}$ against collective i.i.d. attacks. The security parameters $\epsilon_{\mathrm{ET}}, \epsilon_{\mathrm{AT}}, \epsilon_{\mathrm{EC}}$ and $\epsilon_{\mathrm{PA}}$ refer to the energy test, the acceptance test and the classical subroutines handling error-correction and privacy amplification, whereas $\overline{\epsilon}$ is a virtual parameter related to smoothing, that can be chosen freely. Finally, the set $\mathcal{S}^{\mathrm{E\&A}}$ is a result of the testing procedure and given by
\begin{align*}
    \mathcal{S}^{\mathrm{E\&A}} := &\left\{ \sigma \in \mathcal{S}^{\mathrm{ET}}: \right.\\
    &\left.~ \mathrm{Tr}_{E}\left[\sigma\right] \text{ is not $\epsilon_{\mathrm{AT}}$-securely filtered in the AT} \right\},
\end{align*}
where 
\begin{align*}
    \mathcal{S}^{\mathrm{ET}} := &\left\{\sigma \in \mathcal{D}_{\leq}(\mathcal{H}_A \otimes \mathcal{H}_B^{n_c} \otimes \mathcal{H}_E): \textrm{ purification of } \rho_{AB}   \right. \\
    & \left. \land \mathrm{Tr}_{E}\left[\sigma\right] \text{ is not $\epsilon_{\mathrm{ET}}$-securely filtered in the ET} \right\}.
\end{align*}

The remaining optimization can be solved using the two-step method introduced in Refs. \cite{Coles_2016, Winick_2018}, that rewrites the problem into a linearized semi-definite program which, in a first step, is solved iteratively. This is followed by a second step, where the obtained iterative bound is turned into a reliable lower bound using SDP duality theory. The coding was carried out in \textsc{Matlab}\textsuperscript{\textregistered} version R2022a and convex optimization problems were modeled using CVX \cite{cvx1,cvx2}, while we used MOSEK (version 10.0.34) \cite{mosek} for the numerical solution of the semidefinite programs.

\section{Network implementation}\label{sec:NetworkImplementation}

We implemented a DM-CVQKD system supporting a multi-user network topology with two receivers (Bobs), as illustrated in Fig.~\ref{fig:TxPic}. This setup extends our single-user DM-CVQKD system to a small-scale quantum access network using a passive optical splitter. Our system incorporates digital signal processing (DSP) for both modulation and signal recovery and a simplified optical subsystem.

\begin{figure*}[t]
   \centering
        \includegraphics[width=0.8\linewidth]{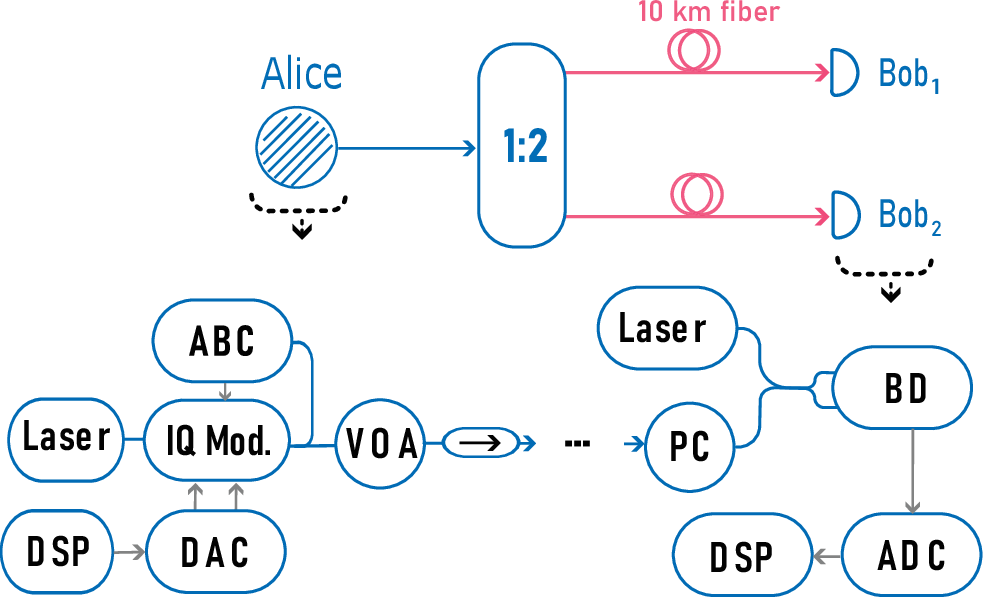}
    \caption{ Schematic of the experimental setup for multi-user DM-CVQKD. Refer to the main text for a detailed description.}
    \label{fig:TxPic}
\end{figure*}

\subsection{Transmitter}

 Alice prepares one of four coherent states using a combination of DSP and optical modules. The real ($q_i$) and imaginary ($p_i$) components of the complex amplitude states were drawn independently from uniform distributions, specifically $q_i, p_i \in \frac{1}{\sqrt{2}} \times \text{Uniform}\{-1,1\}$, at 125~MBaud, using a pseudo-random number generator. These symbols were then upsampled to match the digital-to-analog converter (DAC) sampling rate of 1~GSample/s and pulse-shaped using a root-raised cosine (RRC) filter with a roll-off factor of 0.2. The resulting signal was then frequency-shifted to be centered around  200~MHz and frequency-multiplexed with a 25~MHz pilot tone for carrier phase tracking.

The digital waveform was uploaded to the DAC, which drove an in-phase and quadrature (IQ) modulator to modulate a 1550~nm continuous-wave (CW) laser with 100~Hz linewidth. The bias voltages of the IQ modulator were actively stabilized by an automatic bias controller (ABC). An optical variable attenuator (VOA) was employed to fine-tune the amplitude of the coherent states, and an isolator was inserted to mitigate Trojan horse attacks. Last, the quantum signal was distributed to all users via a quantum broadcasting channel composed of a 1:2 passive optical splitter and two 10 km single-mode fiber (SMF) spools.

\subsection{Receivers}

Each of the two Bobs performed coherent detection using RF heterodyne receivers. The incoming quantum signal was mixed with a local oscillator (LO) at a balanced beamsplitter. Each user employed an independent LO laser with a relative frequency offset of approximately 292~MHz and 300~MHz from Alice’s laser for Bob\(_1\) and Bob\(_2\), respectively. Polarization alignment was achieved using manual polarization controllers to maximize interference visibility at the detectors. Each user utilized a balanced detector (BD) with approximately 250~MHz bandwidth, followed by a 1~GSample/s analog-to-digital converter (ADC) to digitize the detected signal. The ADCs were synchronized to the transmitter’s DAC using a shared 10~MHz reference clock.

The DSP pipeline began with a whitening filter to equalize the detector’s spectral response, followed by carrier phase recovery using a pilot-aided machine learning framework. Temporal synchronization was achieved via cross-correlation with known reference sequences. The quantum symbols were then recovered through matched filtering and downsampling to 125~MBaud.

To calibrate the amplitude of the transmitted states, a back-to-back measurement was performed using a short fiber connection between Alice and one of the Bobs. Based on theoretical optimization, the coherent state amplitude was set to approximately 0.69 to maximize the secure key rate under a trusted receiver efficiency of 68\%. After amplitude calibration, each Bob performed three sequential measurements: transmission of the quantum signal over the broadcasting channel, vacuum noise, and electronic noise.

\begin{figure}
\subfloat[\label{fig:AnaAllUntrusted} Scenario a), i.e., all other Bobs untrusted]{
\includegraphics[width=0.45\textwidth]{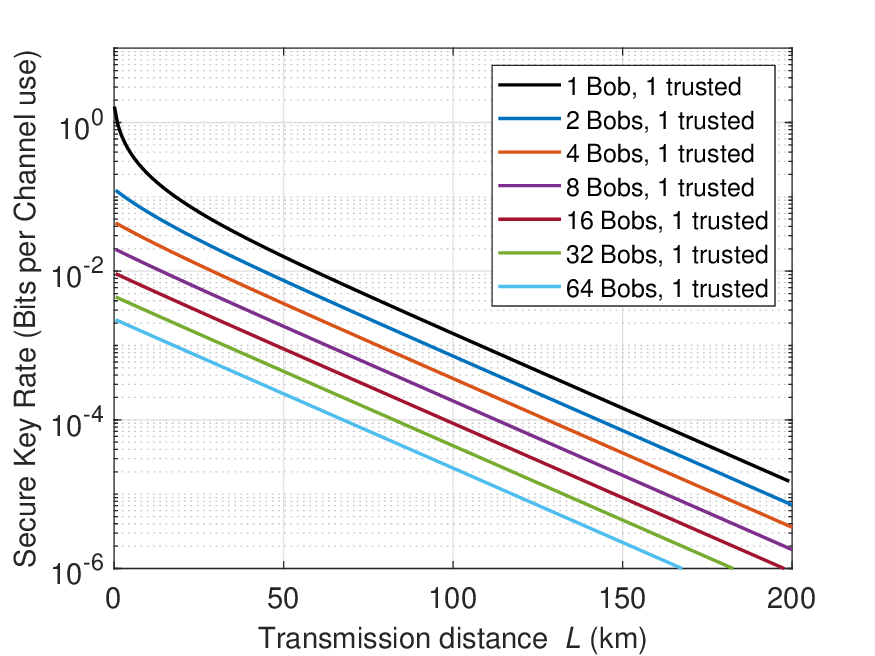}}\\
\subfloat[\label{fig:AnaAllTrusted} Scenario c), i.e., all Bobs trusted]{
\includegraphics[width=0.45\textwidth]{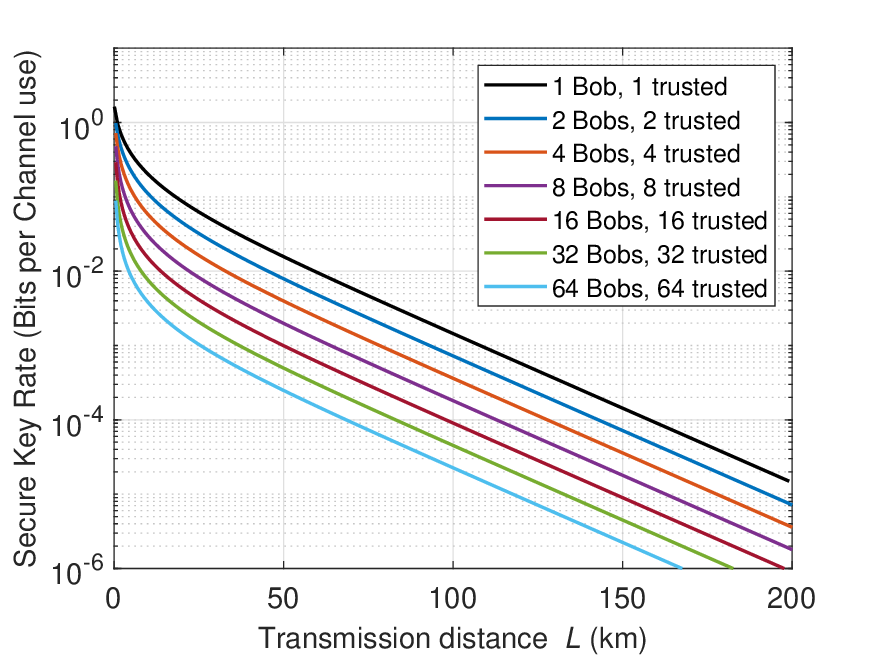}}
\makeatletter\long\def\@ifdim#1#2#3{#2}\makeatother
\caption{QPSK key rate for a single Bob vs transmission distance for the loss-only channel. \label{fig:ResultsKRAna} }
\end{figure}

\section{Post-Processing}\label{sec:Postprocessing}
In order to transform the DSP-processed measurement data into a secure key, a series of classical post-processing steps needs to be performed. To this end, the raw input material is packaged into key blocks containing a fixed number of Quadrature Phase Shift Keying (QPSK) \cite{Lin_2019, Kanitschar_2021} symbols by each peer, Alice and Bob(s), along with a block of metadata describing various characteristics of the key. These keys are subjected to the individual processing steps sequentially in a self-contained software pipeline consisting of several modules, where metadata from previous modules is consumed and additional metadata is provided for subsequent modules. These modules shall be described in detail in the following.

\subsection{Message Authentication}\label{subsec:Authentication}
Most of the modules communicate control messages or process data on a classical channel, such as Ethernet, which can safely be assumed to be error-free, but otherwise completely public and subject to man-in-the-middle attacks. To deal with this situation, the communication needs to be authenticated in an information-theoretically secure manner. To this end, each key metadata is equipped with a data structure prior to post-processing, a so-called crypto-context, containing a fixed-sized hash, both for incoming and outgoing messages. In the following modules, all messages transmitted over the classical channel are hashed onto the crypto-context using a randomly selected universal hash function. After the last processing step, each crypto-context is finalized by hashing a pre-shared secret, and the resulting hashes are openly compared. Any discrepancy reveals tampering of an adversary with the classical communication, resulting in abortion of the key exchange. From each of the successfully processed keys, a small portion is split off and used to replenish the pre-shared secret for subsequent key authentication.

\subsection{Energy Test and Acceptance Test}
The statistical testing procedure can be subdivided into an energy test and an acceptance test. Therefore, Alice selects $k_T$ random rounds and announces their respective labels. Each of the Bobs discloses their measurement results for those rounds over the classical channel. 

The goal of the energy test is to obtain a high probability bound for the maximum dimension of the involved quantum states, which is necessary for the numerical part of the security analysis. For the energy test, the communicating parties compare the amplitude of the announced measurements with a pre-chosen threshold $\beta_{\mathrm{ET}}$ and count the number of rounds that exceed the threshold. For each Bob, the test passes if $\left|\left\{ j \in\{1,...,k_T\}:~ \abs{\frac{q_j + i p_j}{\sqrt{2}}} > \beta_{\text{test}} \right\}\right| \leq \ell_T$, where $q_j$ and $p_j$ are the outcomes of the heterodyne measurement in round $j$.

Next, for the acceptance test, before protocol execution, the communicating parties need to fix an expected channel behaviour, which serves to fix the set of accepted statistics of the protocol. For the present security proof method, for each Bob, they need to fix $\langle \hat{n}_{\beta_i}\rangle$ and  $\langle\hat{n}^2_{\beta_i}\rangle$ for $i\in \{0, 1,2,3 \}$ and compare their observed means, derived based on the disclosed $k_T$ rounds, with the pre-fixed values. The used acceptance test from Ref. \cite{Kanitschar_2023} also allows for non-unique acceptance, where deviations from the expected behaviour can be allowed and included in the security statement, parametrised by a parameter $t$.

For details about the test statements, we refer to Ref. \cite{Kanitschar_2023}, while we refer to the Supplementary Material of Ref. \cite{Hajomer_Kanitschar_2024} for detailed explanations regarding the correct data handling and analysis.

\subsection{Postselection \& Key Mapping}\label{subsec:Postselection}
In a low-SNR scenario, i.e., for long-distance QKD or many Bob instances, the secure key rate can be maximized by discarding certain low-amplitude symbols, thus resulting in an improved SNR. In our implementation, we excluded a circular region around the origin of the quadrature space. The radius of this region is a tuning parameter and is subjected to an optimization procedure. Moreover, the maximum accepted amplitude given by the virtual detection range $M$ is enforced by discarding any symbols exceeding this limit. Both values need to be taken into account in the estimation of the entropy term during parameter estimation (see Eq. \ref{eq:KR_formula_th}).

The quadrature amplitudes of the QPSK symbols are then mapped onto a 2-bit binary string, according to the sign of the amplitudes.

\subsection{Error Correction \& Confirmation}\label{subsec:ErrorCorrection}
Error correction is performed by means of low-density parity-check (LDPC) codes and belief-propagation using the sum-product algorithm. Since the measurement data have been mapped to binary symbols in the previous step, we operate the LDPC decoder in the binary-symmetric channel (BSC) model. For a channel loss of more than 3dB, we operate the decoder in reverse reconciliation, meaning that Bob's key is arbitrarily defined as correct and Alice performs the decoding. We use a set of 22 LDPC codes with native code rates $R_0$ between 11\% and 0.5\%. 
\begin{equation}\label{eq:NativeCodeRate}
    R_0=\frac{n-m}{n},
\end{equation}
with the block length $n$ and the syndrome length $m$.
The code rate can be tuned continuously by selecting the appropriate code and modulating a fixed fraction of symbols with puncturing and shortening \cite{Martinez_2013}, thus obtaining a pre-selected frame error rate according to the estimated QBER. Punctured bits are replaced by random values while shortened bits are set to a fixed value, for example, zero. The effective code rate $R$ thus reads
\begin{equation}\label{eq:EffCodeRate}
    R=\frac{n-m-s}{n-p-s},
\end{equation}
with the number of punctured bits $p$ and shortened bits $s$.
This ensures a smooth efficiency distribution
\begin{equation}\label{eq:LdpcEfficiency}
    \beta=\frac{R}{C(\varepsilon)},
\end{equation}
regardless of the given QBER value $\varepsilon$.
All codes use a fixed block length of $n=512000$ bits and a constant rate modulation of $p+s=8192$ bits. To ensure efficient use of the large input quantum key, it is split up into sub-blocks of size 503808 bits, each of which is extended with the modulated bits in random locations and attempted to be corrected in a maximum of 200 decoder iterations. For successfully corrected sub-blocks, an amount of 
\begin{equation}\label{eq:EcLeakage}
    EC_{leak}=n \times (1-R)
\end{equation}
bits are marked as disclosed, while for failed sub-blocks, we attribute an amount equal to our bound on Eve's knowledge about the key,
\begin{equation}\label{eq:EcLeakage}
    EC_{leak}=n \times \left[ \min_{\rho \in \mathcal{S}^{\mathrm{E\&A}}} H(X|E')_{\Phi(\rho)} - \delta(\bar{\epsilon}) - \Delta(W) \right]
\end{equation}
 disclosed bits. As per the requirement of the underlying security model, no sub-blocks may be discarded, so failed sub-blocks are publicly announced and replaced.

To achieve the requested correctness, an additional verification step is performed. To this end, a randomly selected universal hash is distilled from the key and openly compared between the peers. Again, failed sub-blocks are treated the same way as in the error correction step.

\subsection{Privacy Amplification}\label{subsec:PrivacyAmplification}
Any quantum information Eve may have gained during the physical key exchange, as well as all the classical information leaked during the classical post-processing, needs to be rendered useless in the process of privacy amplification. The final key length is given by Eq. (\ref{eq:KR_formula_exp}), which is achieved with universal hashing using a randomly selected Toeplitz matrix.

\section{Results}
Having established the theoretical security argument and described the experimental setup, we now present our key results. While we reported the key rate formula in Eq. (\ref{eq:KR_formula_th}), for the experimental implementation, it is crucial to clarify how the secure key rates were obtained in practice. Therefore, we reformulate the right-hand side of Eq. (\ref{eq:KR_formula_th}) as follows, taking the actual error-correction leakage, the leaked bits due to privacy amplification, error-verification, and channel authentication into account,

\begin{equation} \label{eq:KR_formula_exp}
\begin{aligned}
    \frac{\ell}{N} \leq \frac{n}{N} &\left[ \min_{\rho \in \mathcal{S}^{\mathrm{E\&A}}} H(X|E')_{\Phi(\rho)} - \delta(\bar{\epsilon}) - \Delta(W) - \mathrm{EC}_{\mathrm{leak}}\right]     \\
    & -\frac{1}{N}\left[ n_{\mathrm{blocks}} b_{\mathrm{EV}} + b_{\mathrm{PA}} + 4 b_{\mathrm{auth}}  \right].
\end{aligned}
\end{equation}
We note that the entropy term is obtained via convex optimisation based on the method explained in Appendix \ref{sec:NumPrfMethod}, while the terms $\Delta(W)$ and $\delta(\bar{\epsilon})$ are corrections to the obtained entropy. The error-correction leakage $\mathrm{EC}_{\mathrm{leak}}$ is obtained directly from the error-correction routine (see Subsection \ref{subsec:ErrorCorrection}) and the leaked hash-lengths due to other classical subroutines are reported in Table~\ref{tab:expParams}. Adding the epsilons of the authentication procedure as well as the epsilon of the random number generator to the security parameter obtained from theory, we obtain a total security epsilon of $\epsilon = \epsilon_{\mathrm{RNG}} + 2 \epsilon_{\mathrm{auth}} + \epsilon_{\mathrm{EC}} + \max\left\{\frac{1}{2}\epsilon_{\mathrm{PA}} + \overline{\epsilon}, \epsilon_{\mathrm{ET}} + \epsilon_{\mathrm{AT}} \right\}$. Which leaves us to comment on the relation of those parameters to the employed subroutines. While the security parameter of the random number generator is directly available from the data sheet, the security parameter of the message authentication routine is a function of the length of the authenticated transcript and the chosen length of the encryption key $b_{\mathrm{auth}}$ and can be bounded as follows \cite{Krzic_2023} 
\begin{equation}
    \epsilon_{\mathrm{auth}} \leq \frac{|C_{\mathrm{transcript}}|}{b_{\mathrm{auth}}} 2^{-b_{\mathrm{auth}}}.
\end{equation}

The failure probability of the error-correction routine can be bounded by a function of the error-verification hash length and the length of the bitstring that was corrected. The bound is given by \cite{Johansson94polyhashing} 
\begin{align}
    \epsilon_{\text{EC}} = 2^{-b_{\mathrm{EV}}} \times \left\lceil \frac{L_\mathrm{LDPC}}{b_{\text{EV}}} \right\rceil \times \left\lceil \frac{n(1-r_{\perp})}{L_\mathrm{LDPC}}\right\rceil.
\end{align}

The parameters $\epsilon_{\mathrm{AT}}$ and $\epsilon_{\mathrm{ET}}$ correspond to the acceptance test and the energy test and bound the failure probabilities of the corresponding tests. For their relation to protocol and implementation parameters, we refer the reader to Theorems 3 and 4 in Ref. \cite{Kanitschar_2023}. The smoothing parameter $\overline{\epsilon}$ is a virtual protocol parameter that can be chosen freely (but influences the size of the correction term $\delta_{\overline{\epsilon}}$). The parameter $\epsilon_{\mathrm{PA}}$ decreases exponentially as a function of the difference between the entropic quantity $H_{\mathrm{min}}^{\overline{\epsilon}}(X|E)_{\rho}$ in the security argument and the length of the final key. We denote this difference as $b_{\mathrm{PA}}$, hence obtain
\begin{equation}
    \epsilon_{\text{PA}} \leq 2^{-\frac{b_{\text{PA}}}{2}}.
\end{equation}

We note that all subprotocol epsilons decrease exponentially and therefore can (in principle) be made arbitrarily small. For the practical implementation, we aimed for a total security parameter of $\epsilon < 10^{-10}$, which led to the choices reported in Table \ref{tab:secParams} and the associated hash lengths reported in Table \ref{tab:expParams}.

One of the main advantages of using discrete modulation rather than Gaussian modulation is a significantly reduced demand on the random number generator, as for all steps we only need a discrete and uniform probability distribution. In the state generation phase, Alice requires $2N$ bits of entropy to choose the modulation for $N$ QPSK symbols. After state transmission, Alice chooses randomly which $k_T$ of the $N$ rounds shall be disclosed publicly by the communicating parties. Therefore, she needs $\lceil \log_2\binom{N}{k_T}\rceil\le\lceil Nh_2(k_T/N)\rceil$ bits of entropy, where we already provide an easy-to-calculate upper-bound as a function of the binary entropy $h_2$. Finally, for the classical postprocessing, Alice chooses a hash function for every error-correction block to verify correctness, which uses $b_{\mathrm{EV}}$ bits of entropy, along with choosing a random hash function for privacy amplification, which consumes $n-1$ bits of entropy.
In total, we upper-bound the number of required random bits as a function of the total number of signals transmitted by
\begin{equation}
    N_{\mathrm{random}} \leq 2N + \lceil N h_2\left( r_T \right)\rceil + b_{\mathrm{EV}}n_{\mathrm{blocks}} + \lceil(1-r_T)N\rceil, 
\end{equation}
where we inserted $r_T = \frac{k_T}{N}$.\\

To illustrate theory curves, we assume standard optical fibers with a loss of $0.2~\mathrm{dB/km}$ connecting Alice to each of the Bobs, and we simulate the Bobs’ observations under the assumption of a Gaussian channel. We emphasize, however, that this assumption is made solely for illustration and is not required by our security analysis. The theoretical arguments apply to an arbitrary number of Bobs, with any subset being potentially trusted (see scenarios a)–d) in Section~\ref{sec:Theory}). For simplicity of presentation, we further assume a symmetrical network topology, where all Bobs are equidistant from Alice. While our security framework is general, our simulations and experimental implementation focus on a quadrature phase-shift keying (QPSK) protocol, in which four coherent states with identical amplitude $|\alpha|$ are prepared and transmitted, differing only in phase by increments of $\frac{\pi}{2}$.

A central question in QKD networks is how many users can be supported at a given transmission distance. To quantify this limit, we use our analytical proof method to evaluate the achievable secure key rate as a function of distance for varying numbers of Bobs. Figure~\ref{fig:ResultsKRAna} shows the asymptotic key rate for $2^n$ Bobs, with $n \in {0,1,2,3,4,5,6}$. In Figure~\ref{fig:AnaAllTrusted}, we consider the optimistic scenario where all Bobs are mutually trusted, while Figure~\ref{fig:AnaAllUntrusted} depicts the opposite case in which a given Bob distrusts all other users. For each fixed distance, we optimized the secure key rate over the coherent state amplitude $\alpha$ in the interval $[0.3,5]$ using the built-in \textsc{Matlab}\textsuperscript{\textregistered} routine \textit{fminbnd} applied to the negative objective function. 

In both scenarios, we observe qualitatively similar behavior at long transmission distances; however, mutual trust among Bobs significantly improves the key rate in the low to medium distance regime. This observation aligns with the structural reduction used to simplify our security analysis (see Section~\ref{sec:ReductionSingleBob}). Although the signal received by the key-generating Bob remains the same for a fixed number of users, the portion intercepted by Eve differs. This can be quantified by the function $r(\eta):=\sqrt{\frac{1-\eta}{1-\frac{1}{N_B}\eta}}$, which compares Eve’s information in both trust scenarios. For example, at $\eta_1 = 0.954$ (corresponding to $1$~km) and $\eta_{100} = 0.01$ (corresponding to $100$~km), and with $N_B = 16$, we find $r(\eta_1) = 0.219$ and $r(\eta_{100}) = 0.995$. Thus, at $100$~km, Eve’s advantage is nearly identical in both cases, while at $1$~km, it differs by a factor of approximately five, explaining the observed gap in key rate.

Since mutual trust among users proved beneficial in the low- to medium-distance regime, we fixed the transmission distance between Alice and each Bob to $10$~km—a realistic scenario for urban or campus-scale networks—and analyzed the secure key rate as a function of the number of users in the presence of both loss and noise, considering both the asymptotic and finite-size regimes.

In Figure~\ref{fig:KRvsNoOfBobs}, the curves for the idealized lossy channel and the general lossy-and-noisy case show qualitatively similar behavior. Although noise slightly reduces the key rates in the general case, both scenarios support networks with up to $32$ Bobs while still achieving secure key rates of at least $10^{-3}$ bits per symbol.

In contrast, the finite-size regime reveals a significantly higher sensitivity to the number of users. At a $10$km transmission distance, the network supports up to $5$ Bobs with non-zero key rates. In the low-noise regime, the finite-size key rate is comparable to the asymptotic case, reaching $6.6 \times 10^{-3}$ bits per symbol for networks with $5$ trusted Bobs. The corresponding total security parameter was chosen to be $\epsilon = 10^{-10}$. However, adding a sixth Bob results in a key rate drop to zero.

We attribute this sharp decline not to a fundamental limitation, but likely to numerical artifacts introduced by the two-step solver used in our analysis. Specifically, in this case, we notice a non-negligible gap between steps 1 and 2 which leads to a premature drop to zero. Nevertheless, the steepness of the decline suggests that, even in an ideal solver scenario, the network is unlikely to support more than 10 Bobs in the finite-size regime.

\begin{figure}
\centering
\includegraphics[width=0.95\columnwidth]{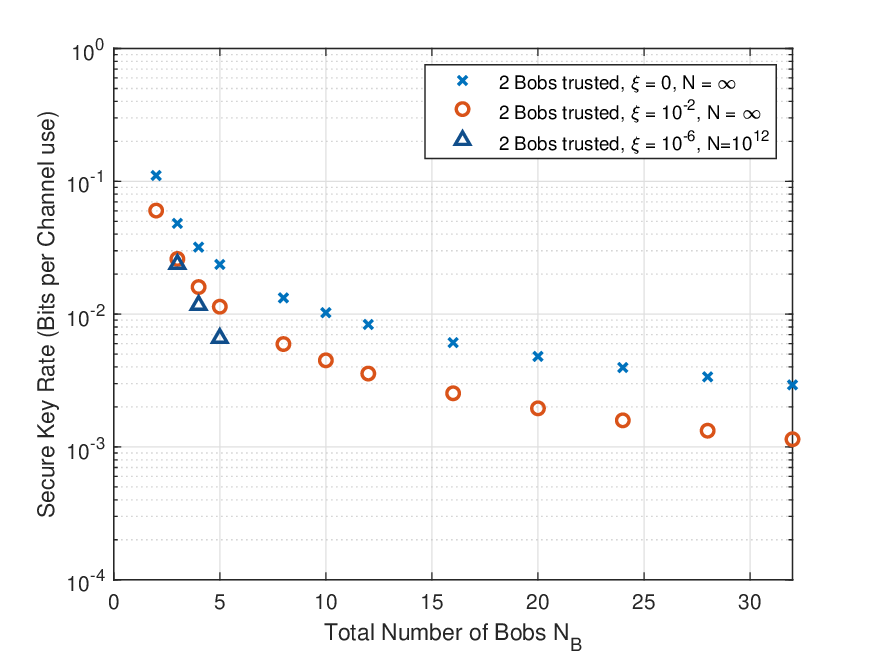}
\caption{Achievable secure key rate vs. total number of Bobs in three different scenarios: loss-only asymptotic, loss and noise asymptotic, loss and noise finite-size, while two Bobs trust each other. The distance between the central node and each Bob was fixed to $10$km and the curves were optimized over the coherent state amplitude $|\alpha|$.\label{fig:KRvsNoOfBobs}}
\end{figure}
Next, we evaluated the network performance in two representative scenarios: one with two mutually trusting Bobs, and another with five Bobs, among whom only two share mutual trust. Figure~\ref{fig:KRvsDistance} presents the secure key rate as a function of transmission distance under an excess noise of $\xi = 10^{-2}$. The solid curves depict the asymptotic key rates for networks with 2 Bobs (blue) and 5 Bobs (orange), while the dashed curves represent the corresponding finite-size results.

\begin{figure}
\centering
\includegraphics[width=0.95\columnwidth]{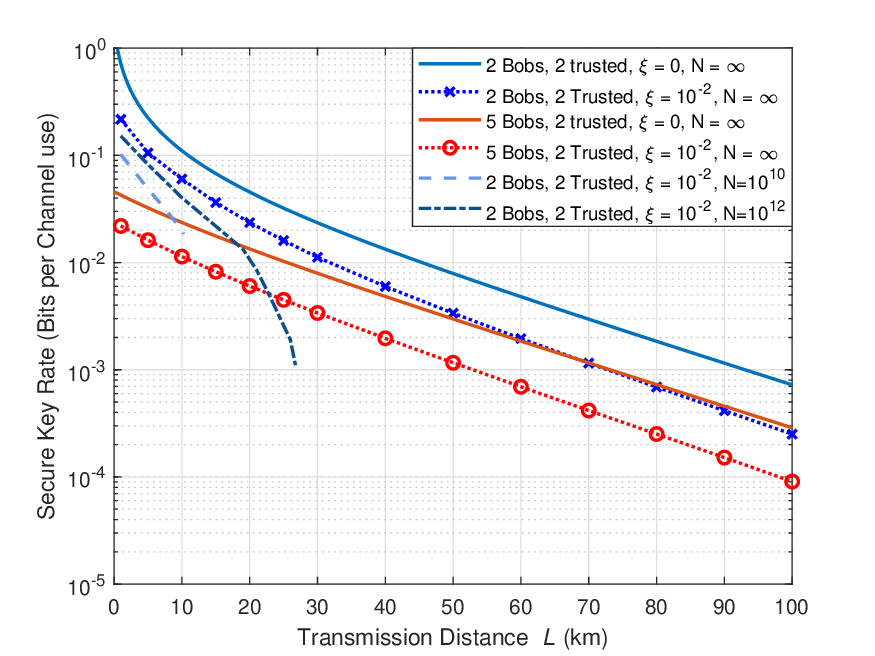}
\caption{Achievable secure key rate vs. transmission distance in the asymptotic (solid curves) and the finite-size regime (dashed curves). Each data point was optimized over the coherent state amplitude $|\alpha|$ as well as over the postselection parameter $\Delta_r$.\label{fig:KRvsDistance}}
\end{figure}

For a block size of $N = 10^{10}$, the finite-size key rate drops to zero slightly beyond 10km. However, this sharp cutoff is not a fundamental limitation but rather a consequence of the numerical issue discussed previously, which can lead to overly pessimistic rate estimates. Increasing the block size to $N = 10^{12}$ mitigates this effect, yielding non-zero key rates up to approximately $26.5$km.

Finally, we implemented our protocol in a passive optical network, where Alice (the central node) was connected to two receivers (the Bobs) via two $10$km optical fibers, as described in Section \ref{sec:NetworkImplementation}. Both Bobs are equipped with imperfect heterodyne detectors. Hardware and Protocol parameters are summarised in Table \ref{tab:generalParams}. We note that the excess noise value is reported only to allow for comparison with other implementations; our work does not assume a Gaussian channel and hence does not rely on the excess noise parameter, but instead uses directly the measured quantities. The postprocessing pipeline was adapted, implemented, and tested for the multi-user scenario, with crucial improvements in the error-correction routines, including the implementation of shortening and puncturing (see Section \ref{subsec:ErrorCorrection} for details). 

We characterized the system and optimized the expected secure key rate over the coherent state amplitude $\alpha$, finding the optimum $\alpha_{\mathrm{opt}} = 0.69$. The protocol was executed for $N=2.3\times 10^9$ rounds in the honest implementation, with further optimisation over the postselection parameter $\Delta_r$. Using $40\%$ of the signals for testing, an energy-test parameter of $\beta_{\mathrm{ET}} = 5.25$, and an optimal postselection parameter of $\Delta_r = 0.575$, we achieved a secure key rate of $2.185 \times 10^{-3}$ bits per symbol, upon acceptance in subsequent protocol executions. The corresponding total security parameter was $\epsilon = 10^{-10}$. This result is marked by a yellow star in Figure \ref{fig:KRvsloss}. At a system repetition rate of $125$MHz, this corresponds to a key rate of  $\approx0.273$Mbit/s.
The observed gap between the theoretical curve and the experimental data in Figure \ref{fig:KRvsDistance} arises primarily due to two factors. First, the experimentally recorded number of signal rounds, $N = 2.3 \times 10^9$, is lower than the theoretically assumed value of $N = 10^{10}$, which naturally reduces the key rate. Second, the theoretical curves were simulated for a lower testing ratio of $r_T = 10\%$, whereas the experiment used $r_T=40\%$, which ultimately yielded better performance in practice.

\begin{figure}
\centering
\includegraphics[width=0.95\columnwidth]{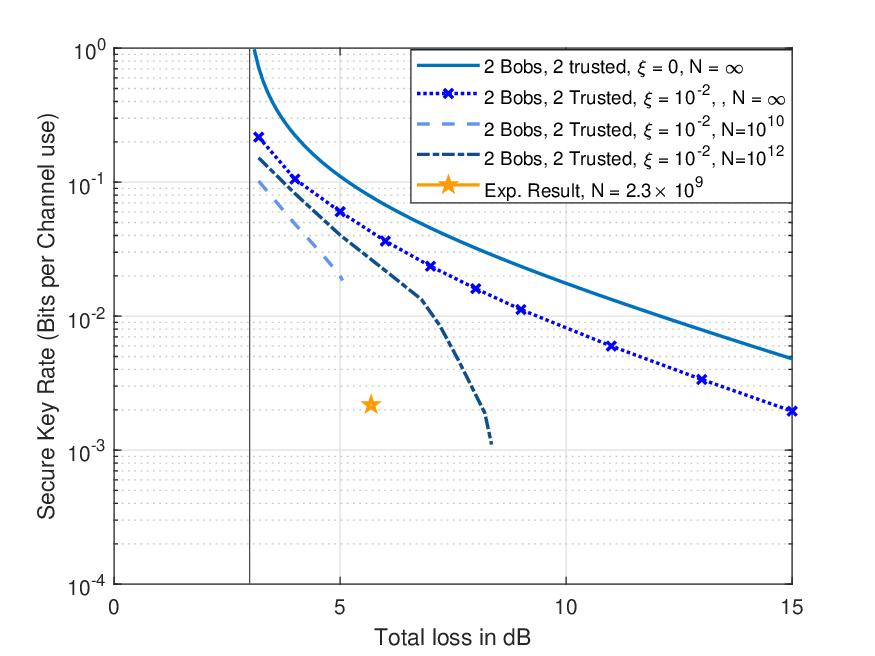}
\caption{Achievable secure key rate vs. total loss for two Bobs. We show the asymptotic noiseless curve (solid blue), the asymptotic curve for $\xi=10^{-2}$ (dotted blue), as well as finite-size curves for $\xi=10^{-2}$ for block-sizes of $N=10^{10}$ (dashed light-blue) and $N=10^{12}$ (dot-dashed dark-blue). Each data point was optimized over the coherent state amplitude $|\alpha|$ as well as over the postselection parameter $\Delta_r$. Additionally, we represent our experimental result, achieved for $N=2.3\times 10^{9}$ with a yellow star.\label{fig:KRvsloss}}
\end{figure}

\begin{table}[t]
\centering
\caption{\textbf{Security parameters of involved (sub-) protocols}.}
\resizebox{0.78\hsize}{!}{
\begin{tabular}{|ccc|}
\hline
\multicolumn{1}{|c|}{\textbf{(Sub-) Routine}}                 & \multicolumn{1}{c|}{\textbf{Symbol}}     & \textbf{Value}                 \\ \hline
\multicolumn{1}{|c|}{\textit{\textbf{QKD Protocol}}}          & \multicolumn{1}{c|}{$\epsilon$}             & $\leq 1 \times 10^{-10}$            \\ \hline
\multicolumn{3}{|c|}{\textit{\textbf{}}}                                                                                                     \\ \hline
\multicolumn{1}{|c|}{\textit{\textbf{Privacy Amplification}}} & \multicolumn{1}{c|}{$\epsilon_{\text{PA}}$} & $4 \times 10^{-15}$ \\ \hline
\multicolumn{1}{|c|}{\textit{\textbf{Error Correction}}}    & \multicolumn{1}{c|}{$\epsilon_{\text{EC}}$} & $ \leq\frac{1}{40} \times 10^{-10}$ \\ \hline
\multicolumn{1}{|c|}{\textit{\textbf{Energy Test}}}           & \multicolumn{1}{c|}{$\epsilon_{\text{ET}}$} & $\frac{1}{10} \times 10^{-10}$ \\ \hline
\multicolumn{1}{|c|}{\textit{\textbf{Acceptance Test}}}       & \multicolumn{1}{c|}{$\epsilon_{\text{AT}}$} & $\frac{8}{10} \times 10^{-10}$ \\ \hline
\multicolumn{1}{|c|}{\textit{\textbf{ Smoothing}}}      & \multicolumn{1}{c|}{$\bar{\epsilon}$}     & $\frac{8}{10} \times 10^{-10}$ \\ \hline
\multicolumn{1}{|c|}{\textit{\textbf{ Random Number Generation}}}      & \multicolumn{1}{c|}{$\epsilon_{\text{RNG}}$}     & $\frac{1}{20} \times 10^{-10}$ \\ \hline
\multicolumn{1}{|c|}{\textit{\textbf{ Authentication}}}      & \multicolumn{1}{c|}{$\epsilon_{\text{auth}}$}     & $<\frac{1}{80} \times 10^{-10}$ \\ \hline
\end{tabular}
}
\label{tab:secParams}
\end{table}

\begin{table}[ht]
\centering
\caption{\textbf{Protocol parameter choices  \& experimental parameters. \label{tab:expParams}}}
\resizebox{0.98\hsize}{!}{
\begin{tabular}{|c|c|c|}
\hline
\textbf{Parameter}                    & \textbf{Symbol}     & \textbf{Value}                                     \\ \hline
\textit{\textbf{Coherent state amplitude}} & $|\alpha|$            & $0.69$                                            \\ \hline
\textit{\textbf{Postselection parameter}} & $\Delta_r$            & $0.575$                                            \\ \hline
\textit{\textbf{Cutoff number}}            & $n_c$               & $7$                                               \\ \hline
\textit{\textbf{Detection limit}}          & $M$                 & $5.25$ (NU)                                             \\ \hline
\textit{\textbf{ET - parameter}}           & $\beta_{\mathrm{ET}}$ & $5.5$                                              \\ \hline
\textit{\textbf{Testing ratio}}           & $r_T$ & $\{40\%\}$                                              \\ \hline
\textit{\textbf{Fraction of outliers}}     & $\frac{l_T}{k_T}$   & $10^{-8}$                                          \\ \hline
\textit{\textbf{Weight}}                   & $w$                 & $\left[1 \times 10^{-7}, 3 \times 10^{-7} \right]$ \\ \hline
\textit{\textbf{$t$-factor}}                   & $t_F$                 & $1$ \\ \hline
\textit{\textbf{Detection efficiency Bob 1}}                   & $\eta_{D,\mathrm{B1} }$                 & $0.68$ \\ \hline
\textit{\textbf{Detection efficiency Bob 2}}                   & $\eta_{D,\mathrm{B2} }$                 & $0.68$ \\ \hline
\textit{\textbf{Electronic noise Bob 1}}                   & $\nu_{\mathrm{el, B1}}$                 & $0.0277$ \\ \hline
\textit{\textbf{Electronic noise Bob 2}}                   & $\nu_{\mathrm{el, B2}}$                 & $0.0291$ \\ \hline
\textit{\textbf{Est. excess noise}}                   & $\xi$                 & $2.7\times 10^{-2}$ \\ \hline
\textit{\textbf{Est. channel transmittance A-B1}}                   & $\eta_{\mathrm{Ch,AB1}}$                 & $0.2247$ \\ \hline
\textit{\textbf{Est. channel transmittance A-B2}}                   & $\eta_{\mathrm{Ch, AB2}}$                 & $0.2665$ \\ \hline
\textit{\textbf{Number of EC blocks}}                   & $n_{\mathrm{blocks}}$                 &4094 \\ \hline
\textit{\textbf{Verification hash length}}                   & $b_{\mathrm{EV}}$                 & $96$ (bit) \\ \hline

\textit{\textbf{Authentication hash length}}                   & $b_{\mathrm{auth}}$                 & $96$ (bit) \\ \hline
\textit{\textbf{Privacy amplification hash length}}                   & $b_{\mathrm{EV}}$                 & $96$ (bit)\\ \hline
\end{tabular}
}
\label{tab:generalParams}
\end{table}

\section{Discussion}\label{sec:Discussion}
In this study, we provide the first security analysis of continuous-variable quantum key distribution with discrete modulation in the multi-user setting, where a single Alice is connected to multiple Bobs, establishing a key with some or even all of them simultaneously. Our analysis showed that the proposed scheme could accommodate a double-digit number of users at urban-area distances (approximately $10$km) in an asymptotic setting and a mid-single-digit number of users in the finite-size setting. Beyond those theory contributions both in the asymptotic and finite-size regime, we implemented a passive optical network consisting of two Bobs connected to Alice and demonstrated secure key rates of $\approx0.886$Mbit/s 

The generated keys can be directly employed to encrypt messages between Alice and multiple Bobs. Moreover, the individually shared keys between Alice and each Bob can be used to distribute random bit strings for conference key agreement. Our work demonstrates the practicality of discrete-modulation continuous-variable QKD (DM-CVQKD) for deployment in quantum optical networks, enabling secure communication among multiple users over short to medium distances, such as those typical in urban or campus environments. Importantly, unlike many theoretical studies that rely on idealised assumptions about detectors, noise, or hardware, our analysis includes realistic parameters. Combined with our practical demonstration of key exchange over a passive optical network - representative of widely deployed telecommunication infrastructure - this work marks a significant step toward secure communication in complex, multi-user network settings. 

In contrast to standard Gaussian-modulated protocols, which were generalized to the multi-user regime recently \cite{Bian_2023, Derkach_2024, Pan_2025}, we used a finite discrete modulation of coherent states. This substantially reduces hardware and postprocessing requirements while remaining fully compatible with modern telecommunication network infrastructure. The analysed star-like network topology can be extended with additional point-to-point or point-to-multipoint links, paving the way toward scalable and secure quantum communication in complex multi-user networks.

 \begin{acknowledgements}
 The authors thank Vladyslav Usenko, Ivan Derkach, and Akash Nag Oruganti for fruitful discussions. This work was funded by the QuantERA II Programme, which received funding from the European Union’s Horizon 2020 research and innovation programme under Grant Agreement No 101017733 and from the Austrian Research Promotion Agency (FFG), project number FO999891361.

\begin{figure}[htb]
\centering
\includegraphics[width=0.4\columnwidth]{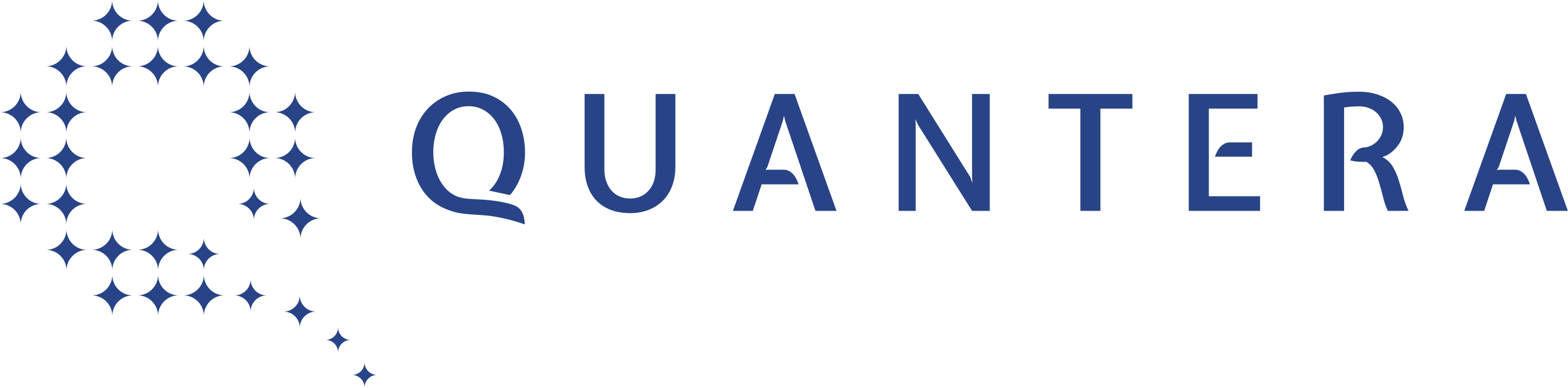}
\end{figure}

\begin{figure}[htb]
\centering
\includegraphics[width=0.8\columnwidth]{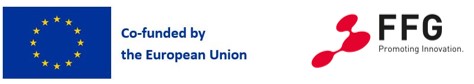}
\end{figure}

 \end{acknowledgements}

 \newpage

\appendix
\onecolumngrid

\section{Protocol Description\label{sec:Protocol}}
We analyze a multi-user version of a general discrete-modulated CV-QKD protocol known from standard (two-user) QKD. Therefore, let us briefly describe the multi-user protocol analyzed in the present work. Let $N_{\text{St}} \in \mathbb{N}$ be the number of different signal states transmitted by a single Alice in the protocol and $N_B \in \mathbb{N}$ the number of Bobs (receivers) in the network. We introduce the short-notations $[m] := \{1,2,..., m\}$ and $[m]_{-1} := \{0,1,...,m-1\}$, where $m \in \mathbb{N}$. Then, the protocol steps read as follows.

\begin{itemize}
    \item[1] \textbf{\textit{State preparation---}} Alice prepares a coherent state $|\alpha\rangle$ with $\alpha \in \{ \alpha_0, ..., \alpha_{N_{St}-1}\}$ with probabilities $p_0,...,p_{N_{\text{St}}-1}$ according to some probability distribution. This state passes a quantum optical network of beamsplitters that split the signal equally into $N_B$ parts and distribute those $N_B$ signals to the $N_B$ Bobs via $N_B$ quantum channels.

    \item[2] \textbf{\textit{Measurement---}} Each of the Bobs receives the quantum signal and performs heterodyne detection, determining the $q$ and $p$ quadrature of the received state. The measurement outcome $y_k^{(i)} \in \mathbb{C}$ of  Bob $i$ is stored in his private register, where $i \in [N_B]$ and $k\in[N]$ is used to label the round.
 \end{itemize}   
 
    Steps 1 and 2 are repeated $N$ times.
    
 \begin{itemize}   
    \item[3] \textbf{\textit{Statistical testing ---}} Once the quantum phase of the protocol is completed, Alice and the $N_B$ Bobs communicate over the classical channel to perform statistical tests. If the test passes, they proceed with the protocol. Otherwise, they abort. Note that - depending on the specific trust assumptions between different users - the tests between Alice and certain Bobs might fail, while the tests between other users might pass.

    \item[4] \textbf{\textit{Reverse reconciliation key map---}} Each of the Bobs performs a reverse reconciliation key map, where he maps his measurement results $y_k^{(i)}$ to discrete values $z_k^{(i)} \in \{0, ..., N_{\text{St}}-1, \perp \}$. This phase allows for postselection, where discarded results are mapped to $\perp$.

    \item[5] \textbf{\textit{Error correction---}} Alice and each of the Bobs communicate over the classical channel to correct their raw keys $\tilde{x}$ and $\tilde{z}^{(i)}$ for $i\in [N_B]$.
    
    \item[6] \textbf{\textit{Privacy amplification---}} Finally, the communicating parties perform privacy amplification. Depending on the trust structure, the details of this phase differ slightly.
\end{itemize}
\subsection{QPSK Modulation and Key Map\label{sec:QPSKProtocol}}
While the presented arguments apply independently of the chosen discrete modulation scheme and the value of $N_{\mathrm{St}}$, for illustration purposes, we present numerical results for a quadrature phase-shift keying (QPSK) protocol, where Alice prepares $N_{\mathrm{St}} = 4$ states (each with probability $1/4$) that are evenly distributed on a circle with radius $|\alpha|$. Then, the key map, which is performed by each of the Bobs in step 4 of the protocol description above, reads as follows, $\forall j \in [N_B]:$

\begin{equation}\label{eq:defRegionsKeyMap}
            \tilde{z}_k^{(j)} = \left\{
\begin{array}{ll}
0 \text{ if } -\frac{\pi}{4} \leq \arg\left(y_k^{(j)}\right) <  \frac{\pi}{4} &\land~ \left|y_k^{(j)}\right| \geq \Delta_r, \\
1 \text{ if } \frac{\pi}{4} \leq \arg\left(y_k^{(j)}\right) <  \frac{3\pi}{4} &\land~  \left|y_k^{(j)}\right| \geq \Delta_r, \\
2 \text{ if } \frac{3\pi}{4} \leq \arg\left(y_k^{(j)}\right) <  \frac{5\pi}{4} &\land~  \left|y_k^{(j)}\right| \geq \Delta_r, \\
3 \text{ if } \frac{5\pi}{4} \leq \arg\left(y_k^{(j)}\right) <  \frac{7\pi}{4} &\land~  \left|y_k^{(j)}\right| \geq \Delta_r, \\
\perp  \textrm{ otherwise,} \\
\end{array}
\right.
        \end{equation}
where $\arg(y)$ denotes the argument of the complex number $y$, and $\Delta_r \geq 0$ is an optional postselection parameter.


\section{The Loss-Only Channel}\label{sec:AnaSecProof}
First, we restrict our considerations to the case of the loss-only channel, i.e., we assume that there is no noise. To simplify the presentation, we assume symmetric channels, i.e., the channels connecting Alice with each of the Bobs are all characterized by the same loss parameter $\eta$. However, our method is not restricted to this case and, at the cost of less instructive expressions, accommodates the general, non-symmetric case as well. We generalize the proof method from Ref.~\cite{Heid_2006} to the multi-user case. Note that \cite{Heid_2006} discusses only a two-state protocol with heterodyne detection; generalizations to protocols with four and higher number of states can be found, for example, in \cite{Lin_2019, Kanitschar_2021, Kanitschar_Thesis_2021}. In the absence of channel noise, it is known that Eve's optimal attack is the generalized beamsplitter attack, i.e., given that Alice prepared a coherent state $\ket{\alpha_x}$, for a channel with loss parameter $\eta$, Bob receives $\ket{\sqrt{\eta}\alpha_x}$, while Eve's share is $\ket{\epsilon_x}:=\ket{\sqrt{1-\eta}\alpha_x}$. Here, $x \in \left[N_{\mathrm{St}}\right]_{-1}$ labels the state Alice prepared. 
In what follows, we are always interested in the secure key rate between Alice and one of the Bobs, given certain trust assumptions on the other Bobs as described in the previous section. To simplify our presentation, we assume that the channels connecting Alice with each of the Bobs are identical, i.e., all channels are described by the same loss-parameter $\eta$. However, we want to emphasize that our argument also works for the general case at the cost of more complicated and less instructive expressions. Therefore, due to symmetry, the key rates between Alice and all of the (trusted) Bobs are the same. We denote the Bob with whom Alice aims to establish a secret key with $B_j$. Then, the asymptotic secure key rate for cases a) - c) is given by the well-known Devetak-Winter formula \cite{Devetak_Winter_2006}:
\begin{equation}\label{eq:DevetakWinterSimple}
    R^{\infty} = \beta I(A:B_j) - \chi(B_j:E),
\end{equation}
where $\beta$ is the reconciliation efficiency, $I(A:B_j)$ denotes the mutual information between Alice and Bob $j$ and $\chi(B_j:E)$ is the Holevo information, a quantity bounding Eve's information about the trusted Bob's quantum state. 

However, for d), where each of the Bobs aims to generate a key with Alice and does not want the other Bobs to hold any information about his key, it is required to take the other Bobs into account. As already discussed in Ref.~\cite{Bian_2023}, the Holevo information must be replaced with the maximum of the mutual information held by Eve or any of the other Bobs. Thus, the modified Devetak-Winter formula reads
\begin{equation}\label{eq:DevetakWinterD}
    R^{\infty} = \beta I(A:B_j) - \max\left\{ \chi(B_j:E), \max_{i\ne j}\{I(B_j,B_i)\}  \right\}.
\end{equation}
In both cases, the main task is to calculate the Holevo quantity, bounding Eve's information for the given situation, since quantifying the correlation between the Bobs for trust scenario d) does not involve Eve's system. Therefore, we aim to describe Eve's system as an orthonormal system.

\subsection{Direct Proof}\label{sec:DirectAnaProof}
 
To illustrate the idea and to point out an issue with a direct analytical approach, let us first focus on scenario c) from the previous section, where Alice trusts all Bobs and aims to establish keys with Bob $i$ and performs the QPSK protocol from Section \ref{sec:QPSKProtocol}. Assume Alice prepared $\ket{\alpha_x}$ in her lab. She splits the state into $N_B$ equal shares and distributes them over $N_B$ noiseless channels characterized by a transmittance $\eta$ to all Bobs. The Bob with whom Alice aims to establish keys then obtains $\ket{\sqrt{\frac{\eta}{N_B}} \alpha_x}$. Eve performs a generalized beamsplitting attack and obtains $\ket{\beta_x} :=\ket{\sqrt{\frac{1-\eta}{N_B}}\alpha_x}$ from each of the channels. The goal is to find a low and finite-dimensional way to represent and diagonalize this state, allowing us to calculate the Holevo quantity. 

For a general phase-shift keying protocol, we can express Eve's state given Alice prepared $\ket{\alpha_x}$ as a superposition of basis states $\ket{b_z}$,
\begin{align*}
    \ket{\Psi_x}_E = \sum_{z=0}^{3} c_z e^{i\arg(\beta_x) z} \ket{b_z}_E,
\end{align*}

 where the coefficients have the form 
\begin{equation}\label{eq:coeffs_c_z}
    c_z =  \sqrt{\sum_{\overset{a_1,...,a_{N_B-1}=0}{\sum_s a_s = z}}^{3} \prod\limits_{s=1}^{N_B} k_{a_s} },
\end{equation}
with
\begin{equation*}
    k_j := e^{-|\beta|^2}\sum_{s=0}^{\infty} \frac{|\beta|^{2(N_{\mathrm{St}}n_s + j)}}{\left(N_{\mathrm{St}} n_s+j\right)!} 
\end{equation*}
for $j\in \{0,1,2,3\}$. Note that we omitted the subscript $x$ in $\beta_x$ because of the symmetry of the considered QPSK protocol. The special form of $|\Psi_x\rangle_E$ allows an analytic calculation of the Holevo quantity $\chi(B:E) = H(\rho_E) - \sum_{j=0}^{3} P(Z=j) H(\rho_{E,j})$, where $\rho_{E,j}$ is Eve's conditional state given that Bob measured the symbol associated with $j$,
\begin{equation*}
    \rho_{E,j} := \sum_{x=0}^{3} \frac{P(X=x, Z=j)}{P(Z=j)} \ketbra{\Psi_x}_E
\end{equation*}
and $\rho_E := \sum_{j=0}^{3} P(Z=j) \rho_{E,j}$. 

The idea to show this statement is to expand the state in Fock representation and group those Fock states in congruence classes $\textrm{ mod } N_{\mathrm{St}}$. We then define $N_{\mathrm{St}}$ basis vectors. As we want the number of basis vectors to be independent of the number of Bobs, we group them in a special way that ensures mutual orthogonality. So, we first aim to write the state Eve branched off in a different way that allows us to regroup the expressions in an advantageous way.
 \begin{align*}
\ket{\Psi_x}_E &= \bigotimes_{b=1}^{N_B} \ket{\beta_x}_{E_b} = e^{-N_B\frac{|\beta_x|^2}{2}} \bigotimes_{b=1}^{N_B}\sum_{n_b=0}^{\infty}\frac{\beta_x^{ n_b}}{\sqrt{n_b!}} \ket{n_b}_{E_b} \\
&= e^{-N_B\frac{|\beta_x|^2}{2}} \sum_{n_1,...,n_{N_B}=0}^{\infty}\bigotimes_{b=1}^{N_B}\frac{\beta_x^{ n_b}}{\sqrt{n_b!}} \ket{n_b}_{E_b} \\
&= e^{-N_B\frac{|\beta_x|^2}{2}} \sum_{a_1,...,a_{N_B}=0}^{N_{\mathrm{St}}-1} \sum_{n_1,...,n_{N_B}=0}^{\infty} \bigotimes\limits_{b=1}^{N_B}\left(\frac{\beta_x^{N_{\mathrm{St}}n_b+a_b}}{\sqrt{\left(N_{\mathrm{St}} n_b+a_b \right)!}} \ket{N_{\mathrm{St}}n_b+a_b}_{E_b}\right)\\
&= e^{-N_B\frac{|\beta_x|^2}{2}} \sum_{z=0}^{N_{\mathrm{St}}-1} \sum_{\overset{a_1,...,a_{N_B}=0}{\sum_j a_j = z \text{ mod } N_{\mathrm{St}}}}^{N_{\textrm{St}}-1} \sum_{n_1,..., n_{N_B}=0}^{\infty} \bigotimes\limits_{b=1}^{N_B}\left(\frac{\beta_x^{N_{\mathrm{St}}n_b+a_b}}{\sqrt{\left(N_{\mathrm{St}} n_b+a_b \right)!}} \ket{N_{\mathrm{St}}n_b+a_b}_{E_b}\right)
\end{align*}   
In the third line we introduce the said congruence classes $0, 1, ..., N_{\mathrm{St}}-1$ in each of the systems, and in the last line, we split this sum into $N_{\mathrm{St}}$ parts satisfying $\sum_j a_j = z \text{  mod  } N_{\mathrm{St}}$ for $z \in \left[N_{\mathrm{St}}\right]_{-1}$. This allows us to define
    \begin{align*}
        \ket{\tilde{b}_z^{(x)}}_E := e^{-N_B\frac{|\beta_x|^2}{2}}   \sum_{\overset{a_1,...,a_{N_B}=0}{\sum_j a_j = z \text{ mod } N_{\mathrm{St}}}}^{N_{\textrm{St}}-1} \sum_{n_1,..., n_{N_B}=0}^{\infty} \bigotimes\limits_{b=1}^{N_B}\left(\frac{\beta_x^{N_{\mathrm{St}}n_b+a_b}}{\sqrt{\left(N_{\mathrm{St}} n_b+a_b \right)!}} \ket{N_{\mathrm{St}}n_b+a_b}_{E_b}\right).
    \end{align*}
Note that $\beta_x^{N_{\mathrm{St}}n_b+a_b} = \left| \beta_x\right|^{N_{\mathrm{St}}n_b+a_b} e^{i \arg(\beta_x) \left(N_{\mathrm{St}}n_b+a_b\right)}$. Given the symmetry of PSK protocols where $\arg(\alpha_x) = \frac{2\pi}{N_{\text{St}}} x$,  $e^{i \arg(\beta_x) \left(N_{\mathrm{St}}n_b\right)} = 1$, which allows us to rewrite the expression as
\begin{align*}
        \ket{\tilde{b}_z^{(x)}}_E := e^{-N_B\frac{|\beta|^2}{2}} e^{i \arg(\beta_x)z}  \sum_{\overset{a_1,...,a_{N_B}=0}{\sum_j a_j = z \text{ mod } N_{\mathrm{St}}}}^{N_{\textrm{St}}-1} \sum_{n_1,..., n_{N_B}=0}^{\infty} \bigotimes\limits_{b=1}^{N_B}\left(\frac{\left|\beta\right|^{N_{\mathrm{St}}n_b+a_b}}{\sqrt{\left(N_{\mathrm{St}} n_b+a_b \right)!}} \ket{N_{\mathrm{St}}n_b+a_b}_{E_b}\right),
    \end{align*}
which shows that $x$ only occurs in the exponential pre-factor. Next, we normalize those vectors. Let us denote 
\begin{equation*}
    k_j := e^{-|\beta|^2}\sum_{s=0}^{\infty} \frac{\left(|\beta|^2\right)^{N_{\mathrm{St}}n_s + j}}{\left(N_{\mathrm{St}} n_s+j\right)!} 
\end{equation*}
for $j\in \left[N_{\mathrm{St}}\right]_{-1}$. Note that due to symmetry implied by phase-shift keying modulation, for all $x$, $k_j$ is independent of $x$. Consequently, to ease notation, we omitted the descriptor $x$. Then the normalization constant $c_z$ for $\ket{\tilde{b}_z^{(x)}}_E$, which is independent of $x$,  reads
\begin{align}
    c_z = \sqrt{\braket{\tilde{b}_z^{(x)}}} = \sqrt{\sum_{\overset{a_1,...,a_{N_B-1}=0}{\sum_s a_s = z}}^{N_{\mathrm{St}}-1} \prod\limits_{s=1}^{N_B} k_{a_s} }, \label{eq:c_z}
\end{align}
and the normalized system is given by
\begin{equation*}
    \ket{b_z^{(x)}}_E := \frac{1}{c_z} \ket{\tilde{b}_z^{(x)}}_E,
\end{equation*}
for $z \in \left[N_{\text{St}}\right]_{-1}$.

This allows us to express Eve's state given Alice prepared $\ket{\alpha_x}$ in a particularly simple form
\begin{align*}
    \ket{\Psi_x}_E = \sum_{z=0}^{N_{\mathrm{St}}-1} c_z e^{i\arg(\beta_x) z} \ket{b_z}_E,
\end{align*}
where we pulled out the $x$-dependent exponential factor from $\ket{b_z^{(x)}}_E$, defining $\ket{b_z}_E := e^{-i\arg(\beta_x) z} \ket{b_z^{(x)}}_E$. Next, we notice the following
\begin{lemma}\label{lem:Lemma1}
    For $ \ket{b_{z}}_E$ as defined above and $z_1,z_2 \in \left[N_{\mathrm{St}}\right]_{-1}$ we have
    $z_1 \neq z_2 ~\Rightarrow ~ \ket{b_{z_1}}_E \perp \ket{b_{z_2}}_E$.
\end{lemma}
\begin{proof}
If $z_1\neq z_2$, then it holds for all $a_1,...,a_{N_B}$ and $d_1,...,d_{N_B} \in \left[N_{\mathrm{St}}\right]_{-1}$ that $z_1 = \sum_j a_j \neq \sum_j d_j = z_2$. Then, we find at least one $j$ such that $a_j \neq d_j$. Thus, in every summand of $\ket{b_{z_1}}_E$ there is (at least) one Fock state $\ket{N_{\mathrm{St}}n_j+a_j}_{E_j}$ that is orthogonal to its counterpart $\ket{N_{\mathrm{St}}n_j+d_j}_{E_j}$. Consequently, the $ \ket{b_{z}}_E$ are mutually orthogonal.
\end{proof}

Since, thanks to Lemma \ref{lem:Lemma1}, those vectors are now not only normalized but orthogonal, the $\left\{\ket{b_z}_E \right\}_z$ form an orthonormal system, hence a basis. This finally allows an analytic calculation of the Holevo quantity $\chi(B:E) = H(\rho_E) - \sum_{j=0}^{N_{\mathrm{St}}-1} P(Z=j) H(\rho_{E,j})$, where $\rho_{E,j}$ is Eve's conditional state given that Bob measured the symbol associated with $j$,
\begin{equation*}
    \rho_{E,j} := \sum_{x=0}^{N_{\mathrm{St}}-1} \frac{P(X=x, Z=j)}{P(Z=j)} \ketbra{\Psi_x}_E
\end{equation*}
and $\rho_E := \sum_{j=0}^{N_{\mathrm{St}}-1} P(Z=j) \rho_{E,j}$. 

For illustration, we consider quadrature phase-shift keying (QPSK) where $N_{\mathrm{St}} = 4$ and assume the most simple nontrivial case $N_B = 2$. Then, the $k_j$ take a particularly simple form, namely
\begin{align*}
    k_0 &= e^{-|\beta|^2} \frac{\cosh\left( |\beta|^2 \right) + \cos\left( |\beta|^2 \right)}{2} \\
    k_1 &= e^{-|\beta|^2} \frac{\sinh\left(|\beta|^2 \right) + \sin\left( |\beta|^2 \right)}{2}\\
    k_2 &= e^{-|\beta|^2} \frac{\cosh\left( |\beta|^2 \right) - \cos\left( |\beta|^2 \right)}{2} \\
    k_3 &= e^{-|\beta|^2} \frac{\sinh\left( |\beta|^2 \right) - \sin\left( |\beta|^2 \right)}{2}
\end{align*}
and the coefficients $c_z$ read
\begin{align*}
    c_0 &= \sqrt{k_0^2 + 2 k_1k_3+k_2^2}\\
    c_1 &= \sqrt{2 k_0 k_1+ 2 k_2 k_3}\\
    c_2 &= \sqrt{k_1^2 + 2 k_0k_2+k_3^2}\\
    c_3 &= \sqrt{2 k_0 k_3+ 2 k_1 k_2}.
\end{align*}
For $N_B = 1$ we recover the single Bob case, already discussed in \cite{Lin_2019, Kanitschar_2021}.

However, as the number of summands in (\ref{eq:coeffs_c_z}) grows exponentially with the number of Bobs involved, i.e., as $4^{N_B-1}$. So, even for QPSK with $N_{\mathrm{St}} = 4$ the calculation of the coefficients $c_z$ becomes already time-consuming (and potentially inaccurate) for double-digit $N_B$. Thus, it is evident that one requires a more efficient method for a larger number of Bobs (and/or a larger number of signal states). The same applies to the case where we trust some but not all of the Bobs, which, in principle, can be done similarly, but suffers from the same computational issue.

\subsection{Reduction to Single Bob Case}\label{sec:ReductionSingleBob}
As the optimal attack in the loss-only scenario is known, this  allows for a reduction from the $N_B$ Bobs to a single Bob. For the following argument, we consider the most general scenario b), where Alice distributes quantum signals to $N_B$ Bobs where only $M_B$ are trusted, while all other $N_B-M_B$ Bobs are assumed to collaborate with Eve. Since cases a) and c) are special cases thereof, they follow from the general case by setting $M_B = 1$ and $M_B = N_B$. We do not restrict our considerations to a specific (discrete) modulation pattern, hence, in general, Alice samples states from an arbitrary discrete set $\{\alpha_0, ..., \alpha_{N_{\mathrm{St}}-1}\}$ according to some discrete distribution. To simplify the argument, we assume that the channels between Alice and each of the Bobs are characterized by the same loss parameter $\eta$. By setting $\eta:= \max_i\{\eta_i\}$, we obtain valid lower bounds for the key rate in the asymmetric case. However, a tight derivation for the general non-symmetric case can be done along similar lines, leading to less instructive explanations and significantly more complicated expressions.

We start by modeling the channel behavior. Since we assume $M_B$ Bobs are trusted, we can attribute the $N_B-M_B$ untrusted outputs of Alice's lab directly to Eve. Thus, we replace the $1:N_B$ beamsplitter in Alice's lab by a $\frac{M_B}{N_B}:\frac{N_B-M_B}{N_B}$ beamsplitter BS 1, where the second output $\ket{\sqrt{1-\frac{M_B}{N_B}} \alpha_x}_{E'}$ is directly routed to Eve. The trusted output $\ket{\sqrt{\frac{M_B}{N_B}} \alpha_x}_{B'}$ hits another (multi-port) beam-splitter BS 2, dividing the input signal into $M_B$ equal shares $\bigotimes_{l=1}^{M_B}\ket{\frac{1}{\sqrt{N_B}}\alpha_x}_{B_l''}$, one for each trusted Bob. Each of those signals propagates through a separate lossy quantum channel, where according to the general beam splitting attack (with a set of beam splitters, BS 3), Eve obtains each a share $\ket{\sqrt{\frac{1-\eta}{N_B}}\alpha_x}_{E_l}$, while each of the Bobs receives $\ket{\sqrt{\frac{\eta}{N_B}}\alpha_x}_{B_l}$. 

The whole procedure can be summarized as,
\begin{align*}
    &\ket{\alpha_x} \stackrel{\textrm{BS 1}}{\mapsto}\\ & \ket{\sqrt{\frac{M_B}{N_B}}\alpha_x}_{A'} \otimes \ket{\sqrt{1-\frac{M_B}{N_B}}\alpha_x}_{E'} \stackrel{\textrm{BS 2}}{\mapsto}\\
    &  \bigotimes\limits_{l=1}^{M_B} \ket{\frac{1}{\sqrt{N_B}} \alpha_x}_{A_l''} \otimes \ket{\sqrt{1-\frac{M_B}{N_B}}\alpha_x}_{E'} \stackrel{\textrm{BS 3}}{\mapsto}\\
    &  \bigotimes\limits_{l=1}^{M_B}\left( \ket{\sqrt{\frac{\eta}{N_B}} \alpha_x}_{B_l}\otimes \ket{\sqrt{\frac{1-\eta}{N_B}}\alpha_x}_{E_l}\right) \otimes \ket{\sqrt{1-\frac{M_B}{N_B}}\alpha_x}_{E'}.
\end{align*}

In the following, we make use of two facts:
(i) We can efficiently solve the single Bob case in which Eve holds a single coherent state, and (ii) unitary transformations on Eve's systems leave her Holevo information invariant. Therefore, we try to unitarily combine all separate coherent states that Eve holds into a single coherent state. We use photon number conservation as a shortcut to obtain the magnitude of the single coherent state of Eve and search for a unitary transformation that combines the outputs $E'$ and all the $E_l$ such that Eve holds $\ket{\sqrt{1-\frac{M_B \eta}{N_B}}\alpha_x}_E$ while having the vacuum state in all other systems. In other words, we ask if there exists a linear optics transformation with complex coefficients $a$ and $b_l$ for $l \in \left[M_B\right]$ such that
\begin{align*}
    a \sqrt{1-\frac{M_B}{N_B}} \alpha_x + \sum_{l=1}^{M_B} b_l \sqrt{\frac{1-\eta}{N_B}} \alpha_x &= \sqrt{1-\frac{M_B\eta}{N_B}} \alpha_x\\
    |a|^2+\sum_{l=1}^{M_B}|b_l|^2 &= 1
\end{align*}
holds. Due to symmetry we anticipate that $\forall l \in \left[M_B\right]:~b_l = \frac{b}{\sqrt{M_B}}$ holds
and obtain the solution
\begin{align*}
    a &= \sqrt{\frac{N_B-M_B}{N_B-M_B\eta}}, \\
    b &= \sqrt{\frac{M_B(1-\eta)}{N_B-M_B\eta}}.
\end{align*}
Hence, we showed that Eve can superpose all her coherent states into a single mode $\ket{\Psi_x}_E=\ket{\sqrt{1-\frac{M_B\eta}{N_B}}\alpha_x}_E$, which is in a tensor product with vacuum states only by quantum optical operations (i.e., beam splitters). Thus, we have effectively simplified the general $N_B$ user scenario to a single Alice to Bob scenario with a modified Eve term, which allows us to calculate a lower bound on the secure key rate without exponentially growing terms like the direct approach in Section~\ref{sec:DirectAnaProof}. In fact, the computational complexity is now independent of the number of Bobs and, therefore, scalable. 
Note that by setting $M_B = 1$, we obtain $\ket{\sqrt{1-\frac{\eta}{N_B}}\alpha_x}_E$, which is the result we expect for the particular case of trust-scenario a), and for $M_B = N_B$, corresponding to trust-scenario c), we obtain $\ket{\sqrt{1-\eta}\alpha}_E$, which as well meets expectations for this particular case. The Holevo quantity can then be calculated along the same lines as in the previous section and we yield the same results as with the direct method, but significantly quicker. The mutual information between Alice and Bob $i$ is calculated similarly to the single-user case, always keeping in mind that the key-generating Bob only receives $\ket{ \sqrt{\frac{\eta}{N_B}}\alpha_x}$. A similar consideration allows for calculating the mutual information between Bob $i$ and other Bobs. Then, the asymptotic secure key rate for scenarios a) - d) follows immediately. We present results for key rates in Figure~\ref{fig:ResultsKRAna}.

\section{The Lossy \& Noisy Channel}\label{sec:NumPrfMethod}
Having dealt with the loss-only case in the previous section, it remains to derive secure key rates for the multi-user scenario for general (i.e., lossy \& noisy) channels. The results we obtained for the loss-only case then serve as a benchmark for low noise parameters in the general case and represent upper bounds on those general key rates. We analyze the general case for arbitrary discrete modulation using the numerical security proof method introduced in Refs.~\cite{Coles_2016,Winick_2018} and applied in Refs.~\cite{Lin_2019, Upadhyaya_2021, Kanitschar_2023} to DM CV-QKD, focussing on the asymptotic case. In what follows, we briefly explain the idea of the security proof method for the single Bob case and show how it is adapted and applied to the multi-user case. We refer to Refs.~\cite{Coles_2016,Winick_2018, Lin_2019, Upadhyaya_2021} for further details.

\subsection{Introduction to the Numerical Method Used}\label{sec:SummaryNumMethod}
We aim to calculate secure key rates for the protocol introduced in Section \ref{sec:Protocol}, where Bob performs heterodyne detection, described by a POVM $\left\{\ketbra{\zeta} \right\}_{\zeta \in \mathbb{C}}$, while in the source-replacement \cite{Curty_2004, Ferenzci_2012} picture, Alice's POVM reads $\{\ketbra{i}\}_{i=0}^{N_{\mathrm{St}}-1}$.

Applying the dimension reduction method \cite{Upadhyaya_2021}, one can rewrite the Devetak-Winter formula as the following optimization problem.
\begin{equation}
    R^{\infty} = \min_{\rho \in \mathcal{S}} H\left(Z|E\right)_{\Phi(\rho)} - \delta^{\mathrm{EC}} - \Delta(W),
\end{equation}
where $\mathcal{S}$ is a subset of the set of all density operators defined by physical requirements and experimental observations, $\Phi$ is a quantum channel describing the protocol steps, $\delta^{\mathrm{EC}}$ denotes the error-correction leakage, $W \geq \Tr{\rho \Pi^{\perp}}$ is the so-called weight of the state outside some finite-dimensional cutoff space $\mathcal{H}^{\text{fin}}$ of $\mathcal{H}$ with $\Pi^{\perp}$ being the projection onto the complement of that space, and
\begin{equation}\label{eq:WeightCorrectionTerm}
    \Delta(W) := \sqrt{W} \log_2(|Z|) + (1+\sqrt{W}) h\left( \frac{\sqrt{W}}{1+\sqrt{W}} \right)
\end{equation}
is the improved weight-dependent correction term from~\cite{Upadhyaya_Thesis_2021}. 
Intuitively, the task is to minimize the key rate over all density matrices compatible with the observations. It turns out that this optimization problem can be rewritten as a semi-definite program with a non-linear objective function. The method introduced in Refs. \cite{Coles_2016, Winick_2018} tackles this optimization problem by a two-step process. In the first step, a linearized version of the problem is solved iteratively, using, for example, the Frank-Wolfe algorithm \cite{Frank_Wolfe_1956}. The output of step 1, which is only an approximate solution to the problem, is then turned into a reliable lower bound in step 2, using SDP duality theory.
Additionally, the influence of numerical constraint-violations on the result is taken into account. The gap between step 1 and step 2 can be used to monitor the result's quality. In almost all cases, the results of step 1 and step 2 (which represent an upper- and lower bound on the secure key rate, respectively) differ only negligibly, indicating tight lower bounds.

This method allows us to consider untrusted ideal and trusted non-ideal detectors, as described in \cite{Lin_2020}. Following the notation there, we denote the trusted detection efficiency of Bob's detectors by $\eta_D$ and the trusted electronic noise by $\nu_{\mathrm{el}}$. For simplicity, we assume that both homodyne detectors, forming one heterodyne detector, are characterized by the same parameters $(\eta_D, \nu_{\mathrm{el}})$. Note that this can be easily achieved by choosing $\eta_D := \min(\eta_{D,1}, \eta_{D,2})$ and $\nu_{\mathrm{el}} := \max(\nu_{\mathrm{el},1}, \nu_{\mathrm{el},2})$, which leads to slightly pessimistic, but secure lower bounds.

The method and its extensions introduced in the present paper build up upon code from Ref. \cite{Kanitschar_2023} (which builds up upon code by Ref. \cite{Upadhyaya_2021}) and was implemented in \textsc{Matlab}\textsuperscript{\textregistered} version R2022a. The convex optimization problems are modeled using CVX \cite{cvx1,cvx2}, and we used the MOSEK solver (version 10.0.34) \cite{mosek} to solve semidefinite programs numerically.

\subsection{Generalization to Multiple Bobs}\label{sec:genMultBobs}
Now, let us generalize this method to the case of multiple Bobs. As untrusted Bobs are assumed to collaborate with Eve, we can immediately reduce the general case to scenario b) by attributing all information that goes to untrusted Bobs directly to Eve. Since in QKD, we conservatively assume that all losses are due to Eve, we can include untrusted Bobs simply by modifying $\eta \mapsto \frac{M_B}{N_B}\eta$. Without loss of generality, to simplify notation, we assume that the first $M_B$ Bobs are trusted. The crucial part for all four scenarios is to quantify Eve's information about the key. Since scenarios a) and c) are special cases of scenario b), and the expression for Eve's information about the key in scenario d) is the same as in scenario b) (see Eqs. (\ref{eq:DevetakWinterSimple}) and (\ref{eq:DevetakWinterD})), we can discuss the main task for all four cases at once.

We analyze the generalized problem within the postprocessing framework of Refs.~\cite{Lin_2019, Upadhyaya_2021}. Therefore, let $\mathcal{S}$ denote the set of density matrices compatible with the experimental observations and the requirement that Alice's reduced density matrix remains unchanged. Let us denote Alice's, the Bobs' and Eve's quantum systems by $A$, $B_i$, and $E$, respectively, whereas Eve's system purifies Alice's and the Bobs' joint density matrix, i.e., $\rho_{AB_1...B_{M_B}E}$ is pure. Furthermore, let $\bar{A}$ and $\bar{B}_i$ denote Alice's and each of the Bobs' private registers, and let us use tildes to denote their public registers, $\tilde{A}$, $\tilde{B_i}$, for $i\in[M_B]$. Let be $j$, the index for the key generating Bob. The values stored in the register are drawn from alphabets $S_a, S_b, S_{\alpha}$ and $S_{\beta}$, respectively. Finally, let $X$ and $Z$ denote Alice's and Bob $j$'s key register. Alice's and the Bobs' measurements can be described by a measurement channel $\mathcal{M}$,
\begin{equation}
    \begin{aligned}
        &\mathcal{M}^{(j)}\left(\rho_{AB_1...B_{M_B}E}\right)\\
        &= \sum_{k,l} \ketbra{a_k}_{\tilde{A}}\otimes\ketbra{\alpha_k}_{\bar{A}}\otimes\ketbra{b_l}_{\tilde{B_j}}\otimes \ketbra{\beta_l}_{\bar{B}_j}\\
        &~~~\otimes\left[\sqrt{P_A^{(k)}} \otimes \sqrt{P_{B_j}^{(l)}} \right] \rho_{AB_TE} \left[\sqrt{P_A^{(k)}} \otimes \sqrt{P_{B_j}^{(l)}} \right],
    \end{aligned}
\end{equation}
where $\left\{P_{B_j}^{(l)}\right\}$ denotes Bob $j$'s POVM, and we implicitly assume that on registers we do not mention explicitly, we perform the identity. The action of the key map can be represented by the following isometry
\begin{equation}
    \begin{aligned}
        V = \sum_{S_a, S_b, S_{\beta}} \ket{g(a,b,\beta)}_Z \otimes \ketbra{a}_{\tilde{A}}\otimes \ketbra{b}_{\tilde{B}_j}\otimes \ketbra{\beta}_{\bar{B}_j},
    \end{aligned}
\end{equation}
where the $g: S_a \times S_b \times S_{\beta} \rightarrow \{0,...,N_{\text{St}}, \perp\}$ is the key map function and $\perp$ is the symbol we use to denote discarded signals. Then, the final classical-quantum state between the key register and Eve reads
\begin{equation}\label{eq:PSmapPhi}
\begin{aligned}
    \tilde{\rho}_{Z[E]} &= V \text{Tr}_{A\bar{A}B_1...B_{j-1}B_j\bar{B}_jB_{j+1}...B_{M_B}}\left[\mathcal{M}\left(\rho_{AB_TE} \right)\right] V^{\dagger}\\
    &=\text{Tr}_{A\bar{A}B_1...B_{j-1}B_j\bar{B}_jB_{j+1}...B_{M_B}}\left[V\mathcal{M}\left(\rho_{AB_TE} \right)V^{\dagger}\right] \\
    =&: \Phi\left( \rho_{AB_1...B_{M_B}E}\right).
\end{aligned}
\end{equation}
The starting point for our examination in scenario b) is the Devetak-Winter formula, which can be rewritten as
\begin{align*}
    R^{\infty} &= I(X:Z) - I(Z:E)
    = H(Z|E)-H(Z|X),
\end{align*}
where the occurring quantities are evaluated on the state processed by the communicating parties. Here, $Z$ and $X$ denote Bob $j$'s and Alice's key string. Since we aim for a lower bound on the secure key rate, we have to minimize this expression over all states compatible with all (trusted) Bobs' observations. We replace the second term with the actual error correction leakage $\delta_{EC}^{\text{leak}}$ and thus are only left with an optimization over the first term. We obtain
\begin{align}
    R^{\infty} = \min_{\mathcal{S}} H(Z|E)_{\tilde{\rho}_{ZE}} - \delta_{EC}^{\text{leak}}.
\end{align}
For scenario d), along similar lines, we can rewrite the key rate formula from Eq. (\ref{eq:DevetakWinterD}) as
\begin{equation}
R^{\infty} = \min\left\{\min_{\mathcal{S}} H(Z|E)_{\tilde{\rho}_{ZE}}, \min_i\{H(Z|B_i) \} \right\} - \delta_{EC}^{\text{leak}}.
\end{equation}

Our main task in what follows is to determine $\min_{\mathcal{S}} H(Z|E)_{\tilde{\rho}_{ZE}}$, which is the same for both cases. Following the ideas in \cite{Coles_2012}, we can simplify $ H(Z|E)_{\tilde{\rho}_{ZE}}$ further. To simplify notation, we extend the short notation $B_T$ for $B_1...B_{M_B}$ to $B_{T\setminus j}$ in case we mean $B_1, ..., B_{j-1}, B_{j+1},...,B_{M_B}$, i.e., to denote all Bobs' registers but $B_j$. Then, note that $\tilde{\rho}_{AB_T} = \sum_{i} Z_i \rho_{AB_T} Z_i$ for $Z_i$ being orthogonal projectors with $\sum_i Z_i = \mathbbm{1}$, which allows us to rewrite the term we want to minimize
\begin{align*}
     &H(Z|E)_{\tilde{\rho}_{ZE}} \\
     =& H(\tilde{\rho}_{ZE}) - H(\rho_E)\\
     =& H(\tilde{\rho}_{AB_T})- H(\rho_{AB_T})\\
     =& - \Tr{\tilde{\rho}_{AB_T} \log_2 \tilde{\rho}_{AB_T} } - H(\rho_{AB_T}) \\
     =&\Tr{-\rho_{AB_T} \log_2 \tilde{\rho}_{AB_T}+ \rho_{AB_T} \log_2 \rho_{AB_T}}\\
     =&D\left(\rho_{AB_T} \right|\!\left| \tilde{\rho}_{AB_T}\right).
\end{align*}
The first equality is the definition of the conditional entropy; for the second equality, we use that $\rho_{AB_TE}$ is pure, and the fourth equality exploits the fact that $\tilde{\rho}_{AB_T}$ is already block-diagonal, hence we can replace the first $\tilde{\rho}_{AB_T}$ by $\rho_{AB_T}$ without changing the result. Finally, for the last equality, we use the definition of the quantum relative entropy. A more detailed argument follows the lines of \cite[Appendix A]{Lin_2019}. This justifies using the numerical method from Refs.~\cite{Coles_2016, Winick_2018}, tailored for the quantum relative entropy. 

Next, let us discuss the optimization problem for the objective protocol in detail. Recall that Alice prepares one out of $N_{\textrm{St}}$ quantum states and sends them to the Bobs, who are equipped with heterodyne detectors. In every testing round, one of the Bobs measures to determine the expected value of certain observables (in our case $\hat{n}_{\beta_i}$ and $\hat{n}_{\beta_i}^2$ as we will discuss later). In key-generation rounds, all but one Bob act passively, while the key-generating Bob's POVM is $\left\{\frac{1}{\pi}\ketbra{\zeta}\right\}_{\zeta \in \mathbb{C}}$. Finally, in the source-replacement picture~\cite{Curty_2004, Ferenzci_2012}, Alice's POVM is given by $\left\{\ketbra{i} \right\}_{i=0}^{N_{\text{St}}-1}$. For key generation, we perform a key map that assigns measurement outcomes lying in a certain region of phase space some logical bit-value. Therefore, we need to combine the POVM elements corresponding to a certain region to obtain the corresponding coarse-grained POVM (see also Ref.~\cite{Lin_2019})
\begin{equation}
R_{B_j}^z = \frac{1}{\pi} \int_{A_z} \ketbra{\zeta} d^2\zeta,
\end{equation}
where $A_z$ labels the region. For the objective QPSK protocol, $z\in\{0,1,2,3\}$ and $A_z$ are wedges in phase space, defined in Eq. (\ref{eq:defRegionsKeyMap}). As outlined in \cite{Upadhyaya_2021}, the key map for our protocol simply copies Bob's private register $\bar{B}$ to $Z$, hence we can relabel $\bar{B}$ to $Z$. Thus, we finally obtain for the postprocessing map $\Phi$, defined in Eq. (\ref{eq:PSmapPhi}) 
\begin{equation}
\begin{aligned}
      &\Phi\left(\rho_{A B_T E}\right) =\\
      &~\sum_{z=0}^{N_{\text{St}}-1}\ket{z}_Z\otimes \text{Tr}_{AB_T }\left[\rho_{AB_TE}\left(R_{B_j} ^{z}\ast\mathbbm{1}_{AB_{T\setminus j} }\right)\right],  
\end{aligned}
\end{equation}
where we introduced the notation $\hat{O}_{B_j} \ast \mathbbm{1}_{AB_{T\setminus j} }:= \mathbbm{1}_{AB_{1,...,j-1} }\otimes \hat{O}_{B_j} \otimes \mathbbm{1}_{AB_{j+1,...,M_B} }$ for the operator $\hat{O}$ acting on the $j$-th Bob system and identity operators acting on Alice' and all other Bobs' systems. Furthermore, according to Ref.~\cite[Appendix A]{Lin_2019}, we omitted redundant registers. 

Similar to the single-Bob case \cite{Upadhyaya_2021}, the expected coherent state Bob receives in the case of a loss-only channel is $\ket{\beta_i} = \hat{D}(\beta_i)\ket{0}$, where $\beta_i := \sqrt{\eta}\alpha_i$ and $\{\beta_i\}_{i=0}^{N_{\text{St}}-1}$ is a set of complex numbers that we use to parametrize basis vectors and observables later on. However, considering lossy \& noisy channels, we expect Bob to receive a displaced thermal state. This motivates the basis of displaced number states $\ket{n_{\beta_i}} = \hat{D}(\beta_i)\ket{n}$ as an efficient choice for a basis for our finite-dimensional subspace. This choice will allow us to describe the received state sufficiently well with a small number of basis states (as in the limit of no noise, only one vector suffices). Since this applies to each Bob, we naturally generalize the choice for the single-Bob case to the multi-Bob case and choose the following projection 
\begin{equation}\label{eq:ProjectionNc}
    \Pi^{N_c}_{\beta_i, M_B} := \bigotimes_{k=1}^{M_B} \sum_{n=0}^{N_c} \ketbra{n_{\beta_i}}.
\end{equation}

This already sets the stage for the generalization of observables used in \cite{Upadhyaya_2021} to the multi-Bob case, as it will turn out to be beneficial if they commute with the projection onto the finite-dimensional state we just chose. We choose the observables to be $\hat{n}_{\beta_i}:=\hat{D}(\beta_i)\hat{n}\hat{D}^{\dagger}(\beta_i)$ and  $\hat{n}^2_{\beta_i}:=\hat{D}(\beta_i)\hat{n}^2\hat{D}^{\dagger}(\beta_i)$ which are easily accessible by heterodnye measurement (for details see \cite{Upadhyaya_Thesis_2021}). Then, the set of observables for our protocol read
\begin{equation}
    \begin{aligned}
        \left\{ \hat{\Gamma}_{i,k}, \hat{\Gamma}^{\text{sq}}_{i,k} \right\}_{i\in \left[N_{\text{St}}\right]_{-1}}^{k\in \left[M_B\right]}.
    \end{aligned}
\end{equation}
with
\begin{equation}\label{eq:GammaHats}
    \begin{aligned}
        \hat{\Gamma}_{i,k}&:= \ketbra{i}_A \otimes \left[\left(\hat{n}_{\beta_i}\right)_{B_k}\ast \mathbbm{1}_{B_{T\setminus k}}\right],\\
        \hat{\Gamma}^{\text{sq}}_{i,k} &:= \ketbra{i}_A \otimes \left[\left(\hat{n}^2_{\beta_i}\right)_{B_k}\ast \mathbbm{1}_{B_{T\setminus k}}\right].
    \end{aligned}
\end{equation}
Note, that they commute with $\Pi^{N_c}_{\beta_i, M_B}$ from Eq. (\ref{eq:ProjectionNc}). 

It remains to determine the weight outside of the cutoff space defined by $\Pi^{N_c}_{\beta_i, M_B}$, where $\overline{\Pi^{N_c}_{\beta_i, M_B}}:= \mathbbm{1}-\Pi^{N_c}_{\beta_i, M_B}$. The generalized version of the dual SDP for $W_i = \Tr{\rho_{B_{T}} \overline{\Pi^{N_c}_{\beta_i, M_B}}}$ given in the proof of Theorem 5 in Ref.\cite{Upadhyaya_2021} reads
\begin{align*}
    &\min_{\vec{y}} y_0 + \sum_{k=1}^{M_B} y_k \langle \hat{\Gamma}_{i,k} \rangle + y_k^{\text{sq}} \langle \hat{\Gamma}^{\text{sq}}_{i,k} \rangle\\
    &\text{s.t.: }\\
    & ~y_0 \mathbbm{1}_{B_{T}} +\sum_{k=1}^{M_B} y_k \hat{\Gamma}_{i,k}  + y_k^{\text{sq}}\hat{\Gamma}^{\text{sq}}_{i,k} - \overline{\Pi^{N_c}_{\beta_i, M_B}} \geq 0\\
    &~\vec{y} \in \mathbb{R}^{2M_B+1}.
\end{align*}
The positivity of the operator on the RHS of the first constraint is equivalent to all eigenvalues of this operator being non-negative. Taking into account the structure of this operator, we observe that $\forall k \in \{1,..., M_B\}:~ y_k + y_k^{\text{sq}} \geq 0$ which implies $y_k \geq -y_k^{\text{sq}}$. Additionally, we identify the following candidates for the smallest eigenvalue
\begin{align*}
    &\lambda(m_1,...,m_{M_B-1}) =y_0 -1+ \sum_{k=1}^{M_B} m_k y_k + m_k^2 y_k^{\text{sq}},
\end{align*}
where $m_1,...,m_{M_B} \in \{0,...,N_c+1\}$, with $m_1 = ... = m_B = N_c+1$ being the worst-case candidate. It follows readily that $y_k = -\frac{1}{M_B N_C(N_C+1)}$ and $y_k^{\text{sq}} = \frac{1}{M_B N_C(N_C+1)}$ together with $y_0 = 0$ solves all conditions on the eigenvalues simultaneously. 
Noting that $\left\langle \left(\hat{n}_{\beta_i}\right)_{B_k} \ast \mathbbm{1}_{B_{T \setminus k}}\right\rangle = \langle \hat{n}_{\beta_i} \rangle$ and $\left\langle \left(\hat{n}^2_{\beta_i}\right)_{B_k} \ast \mathbbm{1}_{B_{T \setminus k}}\right\rangle = \langle \hat{n}^2_{\beta_i} \rangle$ we obtain
\begin{equation}
    W_i = \frac{\langle \hat{n}^2_{\beta_i}  \rangle - \langle \hat{n}_{\beta_i}  \rangle}{N_c(N_c+1)},
\end{equation}
which coincides with the single Bob case. Thus, we obtain for the weight, in accordance with the single Bob case from Ref. \cite{Upadhyaya_2021},
\begin{equation}
   W = \sum_{i=0}^{N_{\textrm{St}}-1} p_i \frac{\langle \hat{n}^2_{\beta_i}  \rangle - \langle \hat{n}_{\beta_i}  \rangle}{N_c(N_c+1)}. 
\end{equation}
Having discussed and specified all parts that differ from the single Bob case, we finally can formulate the relevant optimization problem, in analogy to the single Bob case in Ref.~\cite{Upadhyaya_2021}, but using the improved correction from Ref.~\cite{Upadhyaya_Thesis_2021}. We define the objective function 
\begin{equation}
    f(\rho) := H(Z|E)_{\Phi(\rho)}
\end{equation}
and arrive at the following minimization problem
\begin{equation}
    \begin{aligned}
        \text{min}_{\rho}& f(\rho)\\
        \text{s.t.: }&\\
        & 1-W \leq \Tr{\rho} \leq 1,\\
        & \frac{1}{2} \left|\left|\text{Tr}_{B_{T}}\left[ \rho \right]-\rho_A \right|\right|_1 \leq \sqrt{W},\\
        & \Tr{\rho \frac{1}{p_i}\hat{\Gamma}_{i,k}} \leq \gamma_{i,k},\\
        & \Tr{\rho \frac{1}{p_i}\hat{\Gamma}_{i,k}^{\text{sq}}} \leq \gamma_{i,k}^{\text{sq}},\\
        & \rho \geq 0,
    \end{aligned}
\end{equation}
where $\rho_A = \sum_{m,n=0}^{N_{\text{St}}-1} \sqrt{p_mp_n} \bra{\alpha_n}\ket{\alpha_m} \ket{m}\!\!\bra{n}$, the $\hat{\Gamma}_{i,k}$ and $\hat{\Gamma}_{i,k}^{\text{sq}}$ have been defined in Eq. (\ref{eq:GammaHats}) and $\gamma_{i,k} := \langle \hat{\Gamma}_{i,k} \rangle = \langle \hat{n}_{\beta_i}\rangle$, $\gamma_{i,k}^{\text{sq}} := \langle \hat{\Gamma}_{i,k}^{\text{sq}} \rangle = \langle \hat{n}^2_{\beta_i}\rangle$ for $i\in \left[N_{\text{St}}\right]_{-1}, k\in \left[M_B\right]$ 
This optimization problem can be solved using the numerical two-step algorithm from Refs.~\cite{Coles_2012, Winick_2018}. Denoting the found optimum by $\rho^*$, we finally obtain for the key rate
\begin{equation}
    R^{\infty} = f(\rho^*) - \delta_{EC}^{\text{leak}} - \Delta(W),
\end{equation}
where $\Delta(W)$ is the improved correction term given in Eq.~(\ref{eq:WeightCorrectionTerm}). 

To consider imperfect detectors and/or trusted detection noise, the POVM for the ideal heterodyne measurement must be replaced with the POVM of the non-ideal heterodyne measurement which affects the region-operators in the objective function and the observables in the constraints, as explained in \cite{Lin_2020}.

\subsection{Extension to the finite-size regime}
The next step is the extension to the finite-size regime, which applies the techniques of Refs. \cite{Kanitschar_2023, Kanitschar_Thesis_2022}. Particularly, the key-generating Bob now needs to additionally perform an energy test to rigorously handle the numerical cutoff required for the numerics and perform an acceptance test, which replaces parameter estimation. For details, we refer the reader to Ref. \cite{Kanitschar_2021}. Applying this method, we obtain a composable secure key of length $\ell$,
\begin{equation}
\begin{aligned}
    \frac{\ell}{N} \leq \frac{n}{N} &\left[ \min_{\rho \in \mathcal{S}^{\mathrm{E\&A}}} H(X|E')_{\Phi(\rho)} - \delta(\bar{\epsilon}) - \Delta(W) \right]     \\
    & - \delta_{\mathrm{leak}}^{\mathrm{EC}} - \frac{2}{N} \log_2\left( \frac{1}{\epsilon_{\mathrm{PA}}} \right),
\end{aligned}
\end{equation}
where $\delta^{\mathrm{EC}}_{\mathrm{leak}}$ takes the classical error correction cost into account, $\Delta(W)$ is given in Eq.~(\ref{eq:WeightCorrectionTerm}), $\delta(\epsilon) := 2 \log_2\left( \mathrm{rank}(\rho_X)+3 \right) \sqrt{\frac{\log_2\left(2/\epsilon \right)}{n}}$ and $\mathcal{S}^{\mathrm{E\&A}}$ is defined below, with security parameter $\epsilon_{\mathrm{EC}} + \max\left\{\frac{1}{2}\epsilon_{\mathrm{PA}}+\bar{\epsilon}, \epsilon_{\mathrm{ET}}+\epsilon_{\mathrm{AT}} \right\}$ against collective i.i.d. attacks. The security parameters $\epsilon_{\mathrm{ET}}, \epsilon_{\mathrm{AT}}, \epsilon_{\mathrm{EC}}$ and $\epsilon_{\mathrm{PA}}$ refer to the energy test, the acceptance test and the classical subroutines handling error-correction and privacy amplification, whereas $\overline{\epsilon}$ is a virtual parameter related to smoothing, that can be chosen freely. Finally, the set $\mathcal{S}^{\mathrm{E\&A}}$ is a result of the testing procedure and given by
\begin{align*}
    \mathcal{S}^{\mathrm{E\&A}} := &\left\{ \sigma \in \mathcal{S}^{\mathrm{ET}}: \right.\\
    &\left.~ \mathrm{Tr}_{E}\left[\sigma\right] \text{ is not $\epsilon_{\mathrm{AT}}$-securely filtered in the AT} \right\},
\end{align*}
where 
\begin{align*}
    \mathcal{S}^{\mathrm{ET}} := &\left\{\sigma \in \mathcal{D}_{\leq}(\mathcal{H}_A \otimes \mathcal{H}_B^{n_c} \otimes \mathcal{H}_E): \textrm{ purification of } \rho_{AB}   \right. \\
    & \left. \land \mathrm{Tr}_{E}\left[\sigma\right] \text{ is not $\epsilon_{\mathrm{ET}}$-securely filtered in the ET} \right\}.
\end{align*}


\section{Further Analysis}\label{sec:FurtherResults}
We illustrate our findings for the quadrature phase-shift keying (QPSK) protocol (see Section \ref{sec:QPSKProtocol}). However, we emphasize that this is just a choice for illustration purposes since our findings in the previous sections are general. Furthermore, to ease presentation, we consider the symmetric case. This means the channels connecting Alice with each of the Bobs are all characterized by the same loss parameter $\eta$ and, in case of noisy channels, the excess-noise $\xi$ is distributed evenly among all channels. 

In the following, we assume standard optical fibers with $0.2$dB loss per kilometer in our loss model. We consider the case where Alice and Bob perform reverse reconciliation, in which Alice corrects her errors according to the information she receives from Bob via the classical channel. We assume a constant reconciliation efficiency of $\beta = 95\%$ to allow comparability with existing works on single-user QKD. However, we want to note that achieving a continuous error-correction efficiency over a wide range of transmissivities, hence for different orders of signal-to-noise ratio, is challenging. To date, it is unclear if this can be achieved with current error-correction routines. We refer the reader to Ref. \cite{Leverrier_2023} for an in-depth discussion. We start our conversation with the loss-only case, followed by the general case of lossy \& noisy channels.

\subsection{Key Rates for the Loss-Only Channel} \label{sec:AnaResults}
Note that the $1$ Bob curve corresponds to the single-user case, already known from earlier works \cite{Lin_2019, Kanitschar_2021}. 

\begin{figure}
\centering
\includegraphics[width=0.5\columnwidth]{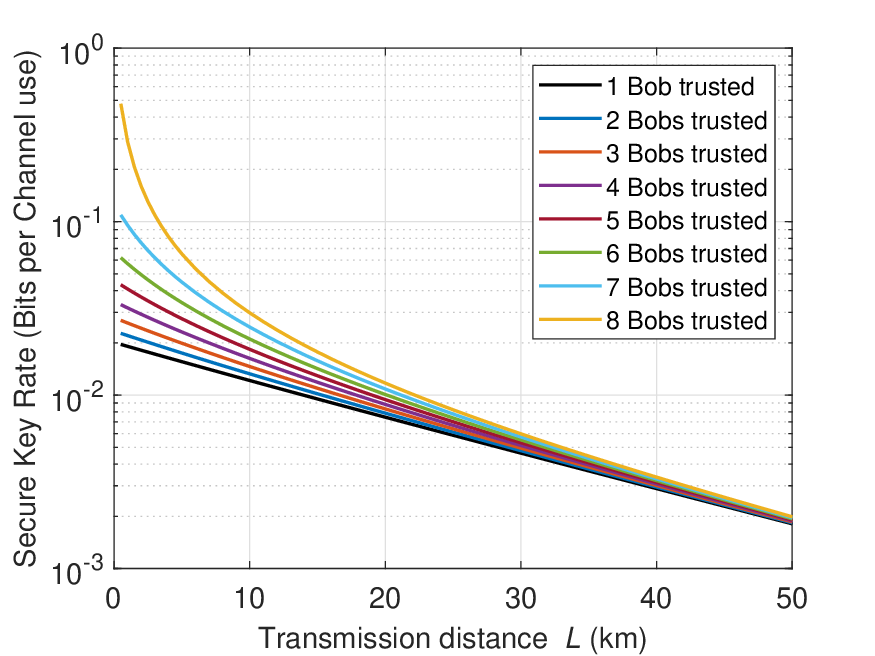}
\caption{QPSK key rate vs transmission distance for the loss-only channel and eight Bobs. Scenario b), i.e., some Bobs are trusted by Alice. \label{fig:Ana_8Bobs}}
\end{figure}

In addition to the examinations in the main part, in Figure \ref{fig:Ana_8Bobs}, we fix the total number of Bobs to $N_B = 8$ and (from bottom to top) trust between $1$ and $8$ Bobs and examine the advantage of trust. The curves where we do not trust all other Bobs ($1$ Bob trusted) are equivalent to the corresponding curves in Figures \ref{fig:AnaAllTrusted} and \ref{fig:AnaAllUntrusted}. We observe very different behaviors for low transmission distances, while the obtained key rate curves are qualitatively the same for medium to large distances. 

These observations already highlight that trusting some or even all Bobs helps to increase the key rate only for low to medium transmission distances. This, however, is well aligned with the main use case of CV-QKD, which distributes secure quantum keys at a high rate in urban areas and campus networks. In particular, trusting some or even all the Bobs can be a practical and realistic use case for campus networks. However, while in such use cases, it might be reasonable to trust some or even all other Bobs not to collaborate with Eve, the Bobs nevertheless might aim for a private key concerning Eve and all other Bobs. We call this scenario, aligning with trust-scenario d). We illustrate this case for $N_B=2,4,6$ and $16$ Bobs in Figure~\ref{fig:Ana_fully_private}, where we optimized $\alpha$ in the interval $[0.3,10]$. We also note that the optimal $\alpha$ for scenario d) differs slightly for low distances. The solid lines correspond to the fully trusted case, while the dotted lines represent the secure key rates in scenario d). We want to highlight that both the solid and the dotted line correspond to the key rate per Bob. While this is equivalent to the total key rate generated by Alice for the trusted case (solid lines), where only one of the Bob generates a key, this differs by a factor of $M_B=N_B$ from the total key rate generated by Alice in scenario d) (dotted lines). This is because, in the latter case, each of the Bobs' keys is completely decoupled from all the other Bobs and, therefore, can be used independently by Alice. 
Considering the key rate per Bob, we observe a noticeable difference in key rate only for very low transmission distances (i.e., in the low-loss case); this small gap even shrinks further for increasing numbers of Bobs. For $N_B>=32$, no discernible difference can be observed through visual inspection; thus, we omitted the corresponding curves. This highlights that obtaining keys according to scenario d) w.r.t. the other Bobs comes at only a tiny cost in key rate per Bob, while it significantly increases the key rate generated by Alice. However, in this picture, we have neglected one effect. Namely, the error-correction phase for scenario d) differs from those in the other scenarios since the syndrome transmitted by one Bob might leak information to the other Bobs. One way to avoid this is by encrypting the syndrome. In the ideal case, where all routines work without errors, this does not come with additional cost since the key is the same length as the syndrome, whose size is subtracted in the key rate formula anyway. If the error correction fails, this procedure uses up the key, thus decreasing the effective key rate. However, since in the asymptotic setting, all subroutines are assumed to work perfectly, we do not consider this error-dependent effect. According to Ref. \cite{Bian_2023}, parts of the lost key can be regained by a recycling procedure that uses the fact that the key used in aborted rounds is `hidden' by the partially random syndrome. While we do not discuss this method in further detail here, we note that it is, in principle, compatible with our work and refer to Ref. \cite{Bian_2023} for further details.

\begin{figure}
\centering
\includegraphics[width=0.4\columnwidth]{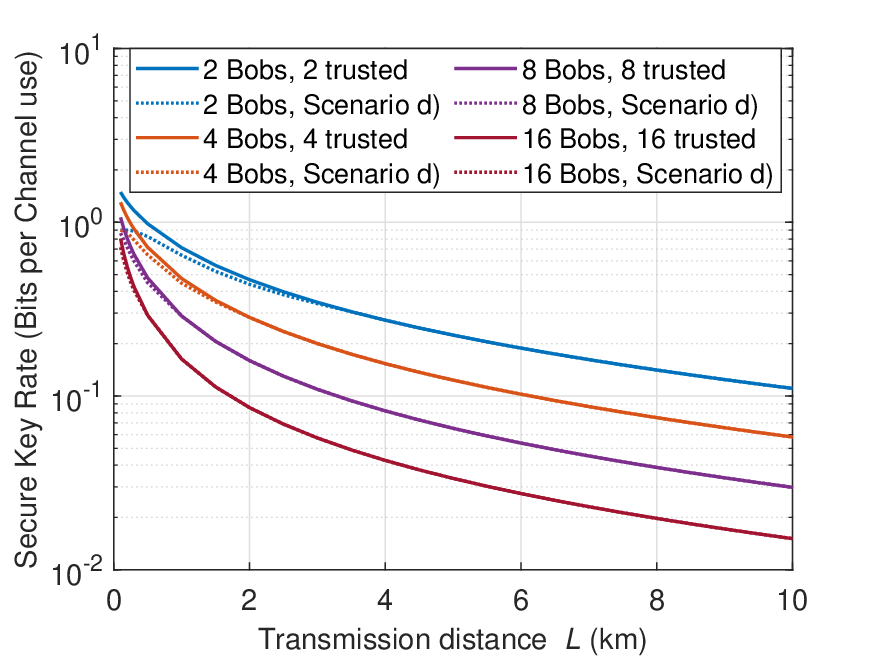}
\caption{QPSK key rate vs transmission distance for the loss-only channel and scenario d), and the trusted scenario c). \label{fig:Ana_fully_private}}
\end{figure}

\begin{figure}
\centering
\includegraphics[width=0.4\columnwidth]{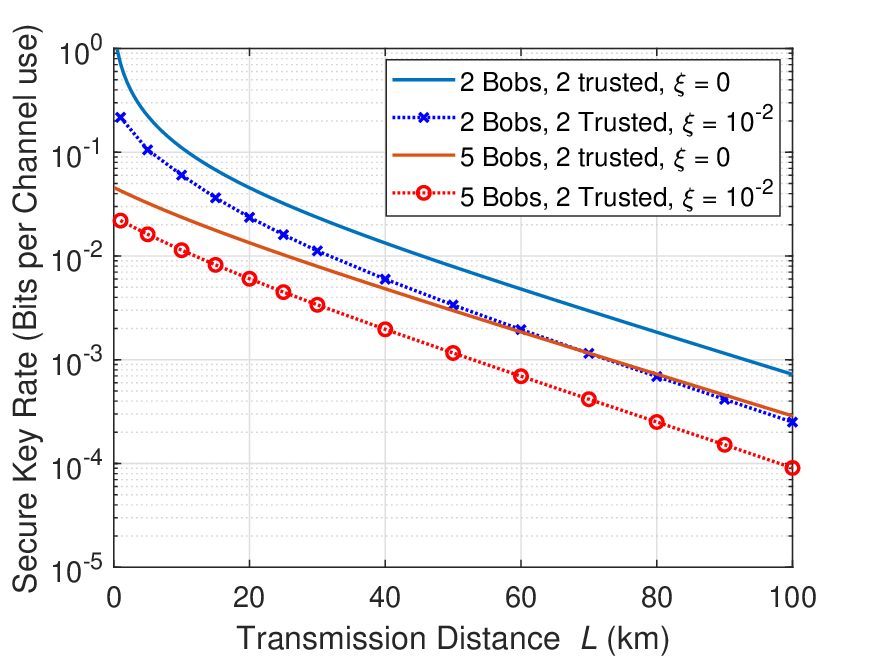}
\caption{QPSK key rate vs transmission distance for the lossy \& noisy channel with $\xi=10^{-2}$. Solid lines are reference curves for the loss-only channel. \label{fig:KR_vs_L_xi_1e-2}}
\end{figure}

\subsection{Key Rates for the Lossy \& Noisy Channel}\label{sec:NumResultsAsympt}
While the assumption of the loss-only channel helped to simplify the security analysis and allowed for analytical solutions, fast evaluations, and qualitative statements, it is far from reality. Thus, we require a security argument for the general lossy \& noisy case. However, the results from the previous section will serve as a helpful benchmark and upper bound for the results in the present section. 

Therefore, we apply the security argument presented in Section \ref{sec:genMultBobs}, using the numerical two-step method by Refs.~\cite{Coles_2016, Winick_2018}, to obtain secure key rates. While our arguments in Section \ref{sec:genMultBobs} apply to a general number of Bobs, calculating secure key rates in this formulation includes solving semi-definite programs. Therefore, computational constraints and limitations need to be taken into account. Let us denote by $n_c$ the largest (displaced) Fock state included in the basis of our finite-dimensional cutoff space. Then the numerical dimension of the problem scales as $\mathcal{O}(N_{\textrm{St}}(n_c+1)^{M_B})$. Thus, it is obvious that to keep the problem numerically feasible, we need to choose a significantly smaller cutoff dimension than for the single Bob case. We found $n_c = 7$ (thus $\textrm{dim}(\mathcal{H}^{\textrm{fin}}) = 8$) being a reasonable compromise. Using the dimension reduction method leads to a higher weight $W$, hence a larger correction term $\Delta(W)$, but after all, still valid lower bounds on the secure key rate. Furthermore, with our current computational resources, we had to restrict our numerical evaluations to a maximum of two trusted Bobs, $M_B = 2$, while keeping the total number of Bobs in principle arbitrary. This is not a fundamental limitation but mainly comes from very RAM-demanding SDP solvers. We expect the number of feasible Bobs can be increased by an improved implementation, using different solvers, better hardware, and/or improved problem formulations that exploit symmetries. Generally, calculating secure key rates for QKD protocols involving high-dimensional systems is known to be RAM-demanding and computationally challenging \cite{Kanitschar_2023_HD}. However, since this work aims to demonstrate our method, arbitrary $N_B$ and fixing $M_B = 2$ is sufficient for our purposes. We demonstrate the soundness of our implementation in Appendix \ref{APDX:Verification_Method} and proceed with our results for an excess noise of $\xi=10^{-2}$ (assumed as preparation noise, inserted into the system from Alice's source). For our plots, for fixed transmission distance $L$, we optimize over the coherent state amplitude $\alpha$ via fine-grained search in steps of $\Delta\alpha = 0.01$. Our security argument for the lossy and noisy case was general and applied to an arbitrary number of Bobs.

In Figure \ref{fig:KR_vs_L_xi_1e-2}, we show the secure key rate for a fixed noise parameter of $\xi = 0.01$ as a function of the transmission distance. As discussed in the previous paragraph, we fix the number of trusted Bobs to $M_B = 2$ and examine a total number of Bobs of $N_B = 2$ (blue curves) and $N_B = 5$ (red curves), although there is no fundamental limit on $N_B$. For comparison, we also plot the corresponding curves for the loss-only scenario ($\xi = 0$) in solid lines. As expected, we observe lowered key rates due to the additional excess noise, but qualitatively similar behavior with loss/transmission distance as for the loss-only case.

\section{Verification of the Numerical Method}\label{APDX:Verification_Method}
This section briefly compares our numerical implementation to the analytical loss-only results. In Figure \ref{fig:Examination_Numerics}, we fix the transmission distance to $L=20$km and the coherent state amplitude to $\alpha=0.87$ (which was the optimal $\alpha$ found for the loss-only case) and plot the obtained secure key rate for decreasing values of excess noise. The key rates converge to the analytical result for $\xi = 0$, with negligible difference for $\xi = 10^{-6}$ and lower. We note that we cannot choose $\xi$ exactly zero for numerical reasons. Next, in Figure \ref{fig:KR_VS_L_Loss_Only}, we fix the excess noise to $\xi = 10^{-6}$ and plot the secure key rate for 2 Bobs where both are trusted and 3 Bobs where 2 of them are trusted for various distances in the interval $[0\textrm{km}, 100\textrm{km}]$ and optimize over $\alpha$ in steps of $\Delta\alpha = 0.01$ around the analytical optimum. Again, we compare our numerical results with the analytical result for the loss-only case and observe excellent accordance. Both Figures show the excellent alignment of our numerical results for the lossy \& noisy case with the analytical loss-only key rates.
\begin{figure}
\subfloat[\label{fig:Examination_Numerics} Key rate vs excess noise $\xi$ compared to loss-only channel for $2$ Bobs where both are trusted.]{
\includegraphics[width=0.40\textwidth]{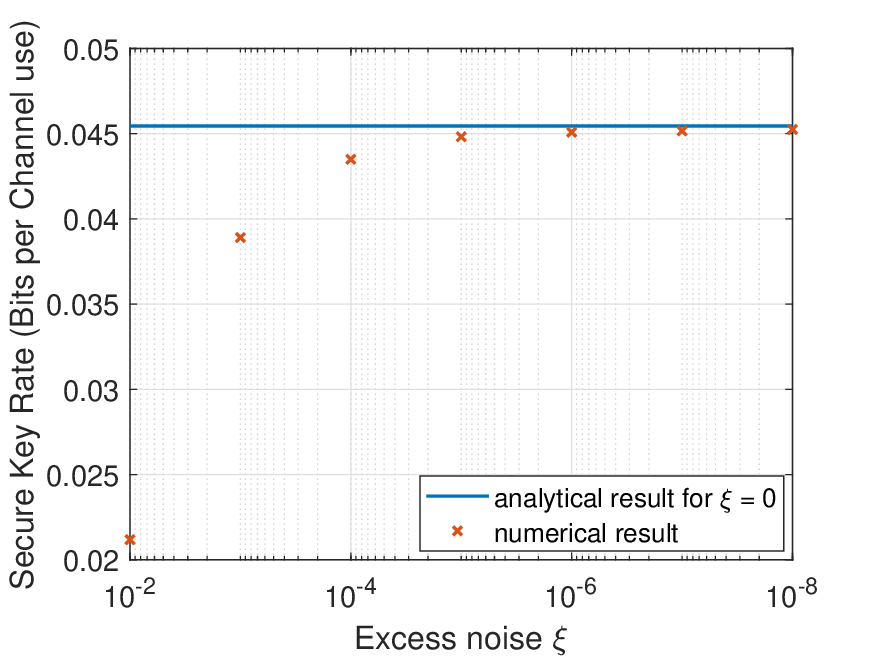}}
\subfloat[\label{fig:KR_VS_L_Loss_Only} Key rate vs transmission distance for $\xi=10^{-6}$ compared to loss-only channel.]{
\includegraphics[width=0.40\textwidth]{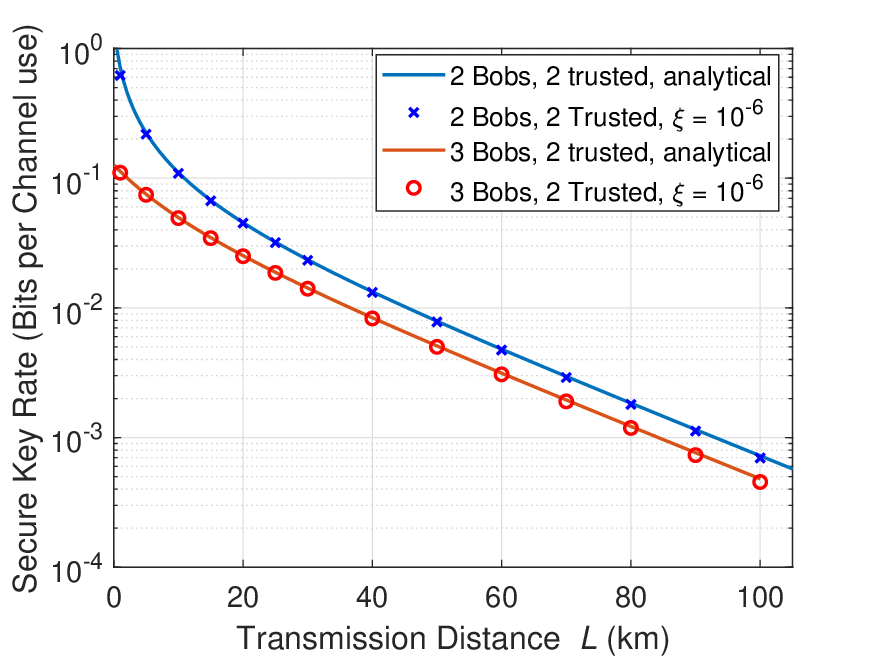}}
\makeatletter\long\def\@ifdim#1#2#3{#2}\makeatother
\caption{Comparison of the numerical method to analytical benchmark for a QPSK protocol. }
\end{figure}

\twocolumngrid
\bibliography{Bibliography}

\providecommand{\noopsort}[1]{}\providecommand{\singleletter}[1]{#1}%
\begin{thebibliography}{53}%
\makeatletter
\providecommand \@ifxundefined [1]{%
 \@ifx{#1\undefined}
}%
\providecommand \@ifnum [1]{%
 \ifnum #1\expandafter \@firstoftwo
 \else \expandafter \@secondoftwo
 \fi
}%
\providecommand \@ifx [1]{%
 \ifx #1\expandafter \@firstoftwo
 \else \expandafter \@secondoftwo
 \fi
}%
\providecommand \natexlab [1]{#1}%
\providecommand \enquote  [1]{``#1''}%
\providecommand \bibnamefont  [1]{#1}%
\providecommand \bibfnamefont [1]{#1}%
\providecommand \citenamefont [1]{#1}%
\providecommand \href@noop [0]{\@secondoftwo}%
\providecommand \href [0]{\begingroup \@sanitize@url \@href}%
\providecommand \@href[1]{\@@startlink{#1}\@@href}%
\providecommand \@@href[1]{\endgroup#1\@@endlink}%
\providecommand \@sanitize@url [0]{\catcode `\\12\catcode `\$12\catcode
  `\&12\catcode `\#12\catcode `\^12\catcode `\_12\catcode `\%12\relax}%
\providecommand \@@startlink[1]{}%
\providecommand \@@endlink[0]{}%
\providecommand \url  [0]{\begingroup\@sanitize@url \@url }%
\providecommand \@url [1]{\endgroup\@href {#1}{\urlprefix }}%
\providecommand \urlprefix  [0]{URL }%
\providecommand \Eprint [0]{\href }%
\providecommand \doibase [0]{https://doi.org/}%
\providecommand \selectlanguage [0]{\@gobble}%
\providecommand \bibinfo  [0]{\@secondoftwo}%
\providecommand \bibfield  [0]{\@secondoftwo}%
\providecommand \translation [1]{[#1]}%
\providecommand \BibitemOpen [0]{}%
\providecommand \bibitemStop [0]{}%
\providecommand \bibitemNoStop [0]{.\EOS\space}%
\providecommand \EOS [0]{\spacefactor3000\relax}%
\providecommand \BibitemShut  [1]{\csname bibitem#1\endcsname}%
\let\auto@bib@innerbib\@empty
\bibitem [{\citenamefont {Bennett}\ and\ \citenamefont
  {Brassard}(1984)}]{Bennett_Brassard_1984}%
  \BibitemOpen
  \bibfield  {author} {\bibinfo {author} {\bibfnamefont {C.~H.}\ \bibnamefont
  {Bennett}}\ and\ \bibinfo {author} {\bibfnamefont {G.}~\bibnamefont
  {Brassard}},\ }\bibfield  {title} {\bibinfo {title} {{Quantum cryptography:
  Public key distribution and coin tossing}},\ }in\ \href@noop {} {\emph
  {\bibinfo {booktitle} {Proceedings of IEEE International Conference on
  Computers, Systems, and Signal Processing}}}\ (\bibinfo  {publisher} {IEEE},\
  \bibinfo {address} {India},\ \bibinfo {year} {1984})\ p.\ \bibinfo {pages}
  {175}\BibitemShut {NoStop}%
\bibitem [{\citenamefont {Ekert}(1991)}]{Ekert_1991}%
  \BibitemOpen
  \bibfield  {author} {\bibinfo {author} {\bibfnamefont {A.~K.}\ \bibnamefont
  {Ekert}},\ }\bibfield  {title} {\bibinfo {title} {{Quantum cryptography based
  on Bell's theorem}},\ }\href {https://doi.org/10.1103/PhysRevLett.67.661}
  {\bibfield  {journal} {\bibinfo  {journal} {Phys. Rev. Lett.}\ }\textbf
  {\bibinfo {volume} {67}},\ \bibinfo {pages} {661} (\bibinfo {year}
  {1991})}\BibitemShut {NoStop}%
\bibitem [{\citenamefont {Ralph}(1999)}]{Ralph_1999}%
  \BibitemOpen
  \bibfield  {author} {\bibinfo {author} {\bibfnamefont {T.~C.}\ \bibnamefont
  {Ralph}},\ }\bibfield  {title} {\bibinfo {title} {{Continuous variable
  quantum cryptography}},\ }\href {https://doi.org/10.1103/PhysRevA.61.010303}
  {\bibfield  {journal} {\bibinfo  {journal} {Phys. Rev. A}\ }\textbf {\bibinfo
  {volume} {61}},\ \bibinfo {pages} {010303} (\bibinfo {year}
  {1999})}\BibitemShut {NoStop}%
\bibitem [{\citenamefont {Usenko}\ \emph {et~al.}(2025)\citenamefont {Usenko},
  \citenamefont {Acín}, \citenamefont {Alléaume}, \citenamefont {Andersen},
  \citenamefont {Diamanti}, \citenamefont {Gehring}, \citenamefont {Hajomer},
  \citenamefont {Kanitschar}, \citenamefont {Pacher}, \citenamefont
  {Pirandola},\ and\ \citenamefont {Pruneri}}]{Usenko_2025}%
  \BibitemOpen
  \bibfield  {author} {\bibinfo {author} {\bibfnamefont {V.~C.}\ \bibnamefont
  {Usenko}}, \bibinfo {author} {\bibfnamefont {A.}~\bibnamefont {Acín}},
  \bibinfo {author} {\bibfnamefont {R.}~\bibnamefont {Alléaume}}, \bibinfo
  {author} {\bibfnamefont {U.~L.}\ \bibnamefont {Andersen}}, \bibinfo {author}
  {\bibfnamefont {E.}~\bibnamefont {Diamanti}}, \bibinfo {author}
  {\bibfnamefont {T.}~\bibnamefont {Gehring}}, \bibinfo {author} {\bibfnamefont
  {A.~A.~E.}\ \bibnamefont {Hajomer}}, \bibinfo {author} {\bibfnamefont
  {F.}~\bibnamefont {Kanitschar}}, \bibinfo {author} {\bibfnamefont
  {C.}~\bibnamefont {Pacher}}, \bibinfo {author} {\bibfnamefont
  {S.}~\bibnamefont {Pirandola}},\ and\ \bibinfo {author} {\bibfnamefont
  {V.}~\bibnamefont {Pruneri}},\ }\href {https://arxiv.org/abs/2501.12801}
  {\bibinfo {title} {Continuous-variable quantum communication}} (\bibinfo
  {year} {2025}),\ \Eprint {https://arxiv.org/abs/2501.12801} {arXiv:2501.12801
  [quant-ph]} \BibitemShut {NoStop}%
\bibitem [{\citenamefont {Cerf}\ \emph {et~al.}(2001)\citenamefont {Cerf},
  \citenamefont {L{\'{e}}vy},\ and\ \citenamefont {Assche}}]{Cerf_2001}%
  \BibitemOpen
  \bibfield  {author} {\bibinfo {author} {\bibfnamefont {N.~J.}\ \bibnamefont
  {Cerf}}, \bibinfo {author} {\bibfnamefont {M.}~\bibnamefont {L{\'{e}}vy}},\
  and\ \bibinfo {author} {\bibfnamefont {G.~V.}\ \bibnamefont {Assche}},\
  }\bibfield  {title} {\bibinfo {title} {{Quantum distribution of Gaussian keys
  using squeezed states}},\ }\href {https://doi.org/10.1103/physreva.63.052311}
  {\bibfield  {journal} {\bibinfo  {journal} {Phys. Rev. A}\ }\textbf {\bibinfo
  {volume} {63}},\ \bibinfo {pages} {052311} (\bibinfo {year}
  {2001})}\BibitemShut {NoStop}%
\bibitem [{\citenamefont {Grosshans}\ \emph {et~al.}(2003)\citenamefont
  {Grosshans}, \citenamefont {Wenger}, \citenamefont {Tualle-Brouri},
  \citenamefont {Grangier}, \citenamefont {Assche},\ and\ \citenamefont
  {Cerf}}]{Grosshans_2002}%
  \BibitemOpen
  \bibfield  {author} {\bibinfo {author} {\bibfnamefont {F.}~\bibnamefont
  {Grosshans}}, \bibinfo {author} {\bibfnamefont {J.}~\bibnamefont {Wenger}},
  \bibinfo {author} {\bibfnamefont {R.}~\bibnamefont {Tualle-Brouri}}, \bibinfo
  {author} {\bibfnamefont {P.}~\bibnamefont {Grangier}}, \bibinfo {author}
  {\bibfnamefont {G.}~\bibnamefont {Assche}},\ and\ \bibinfo {author}
  {\bibfnamefont {N.}~\bibnamefont {Cerf}},\ }\bibfield  {title} {\bibinfo
  {title} {{Quantum key distribution using Gaussian-modulated coherent
  states}},\ }\href {https://doi.org/10.1109/EQEC.2003.1314285} {\bibfield
  {journal} {\bibinfo  {journal} {Nature}\ }\textbf {\bibinfo {volume} {421}},\
  \bibinfo {pages} {238–241} (\bibinfo {year} {2003})}\BibitemShut {NoStop}%
\bibitem [{\citenamefont {Grosshans}\ and\ \citenamefont
  {Grangier}(2002)}]{Grangier_2002}%
  \BibitemOpen
  \bibfield  {author} {\bibinfo {author} {\bibfnamefont {F.}~\bibnamefont
  {Grosshans}}\ and\ \bibinfo {author} {\bibfnamefont {P.}~\bibnamefont
  {Grangier}},\ }\bibfield  {title} {\bibinfo {title} {{Continuous Variable
  Quantum Cryptography Using Coherent States}},\ }\href
  {https://doi.org/10.1103/PhysRevLett.88.057902} {\bibfield  {journal}
  {\bibinfo  {journal} {Phys. Rev. Lett.}\ }\textbf {\bibinfo {volume} {88}},\
  \bibinfo {pages} {057902} (\bibinfo {year} {2002})}\BibitemShut {NoStop}%
\bibitem [{\citenamefont {Silberhorn}\ \emph {et~al.}(2002)\citenamefont
  {Silberhorn}, \citenamefont {Ralph}, \citenamefont {L\"utkenhaus},\ and\
  \citenamefont {Leuchs}}]{Silberhorn_2002}%
  \BibitemOpen
  \bibfield  {author} {\bibinfo {author} {\bibfnamefont {C.}~\bibnamefont
  {Silberhorn}}, \bibinfo {author} {\bibfnamefont {T.~C.}\ \bibnamefont
  {Ralph}}, \bibinfo {author} {\bibfnamefont {N.}~\bibnamefont
  {L\"utkenhaus}},\ and\ \bibinfo {author} {\bibfnamefont {G.}~\bibnamefont
  {Leuchs}},\ }\bibfield  {title} {\bibinfo {title} {{Continuous Variable
  Quantum Cryptography: Beating the 3 dB Loss Limit}},\ }\href
  {https://doi.org/10.1103/PhysRevLett.89.167901} {\bibfield  {journal}
  {\bibinfo  {journal} {Phys. Rev. Lett.}\ }\textbf {\bibinfo {volume} {89}},\
  \bibinfo {pages} {167901} (\bibinfo {year} {2002})}\BibitemShut {NoStop}%
\bibitem [{\citenamefont {Navascu\'es}\ \emph {et~al.}(2006)\citenamefont
  {Navascu\'es}, \citenamefont {Grosshans},\ and\ \citenamefont
  {Ac\'{\i}n}}]{Navascues_2006}%
  \BibitemOpen
  \bibfield  {author} {\bibinfo {author} {\bibfnamefont {M.}~\bibnamefont
  {Navascu\'es}}, \bibinfo {author} {\bibfnamefont {F.}~\bibnamefont
  {Grosshans}},\ and\ \bibinfo {author} {\bibfnamefont {A.}~\bibnamefont
  {Ac\'{\i}n}},\ }\bibfield  {title} {\bibinfo {title} {{Optimality of Gaussian
  Attacks in Continuous-Variable Quantum Cryptography}},\ }\href
  {https://doi.org/10.1103/PhysRevLett.97.190502} {\bibfield  {journal}
  {\bibinfo  {journal} {Phys. Rev. Lett.}\ }\textbf {\bibinfo {volume} {97}},\
  \bibinfo {pages} {190502} (\bibinfo {year} {2006})}\BibitemShut {NoStop}%
\bibitem [{\citenamefont {Garc\'{\i}a-Patr\'on}\ and\ \citenamefont
  {Cerf}(2006)}]{Garcia_2006}%
  \BibitemOpen
  \bibfield  {author} {\bibinfo {author} {\bibfnamefont {R.}~\bibnamefont
  {Garc\'{\i}a-Patr\'on}}\ and\ \bibinfo {author} {\bibfnamefont {N.~J.}\
  \bibnamefont {Cerf}},\ }\bibfield  {title} {\bibinfo {title} {{Unconditional
  Optimality of Gaussian Attacks against Continuous-Variable Quantum Key
  Distribution}},\ }\href {https://doi.org/10.1103/PhysRevLett.97.190503}
  {\bibfield  {journal} {\bibinfo  {journal} {Phys. Rev. Lett.}\ }\textbf
  {\bibinfo {volume} {97}},\ \bibinfo {pages} {190503} (\bibinfo {year}
  {2006})}\BibitemShut {NoStop}%
\bibitem [{\citenamefont {Diamanti}\ and\ \citenamefont
  {Leverrier}(2015)}]{Diamanti_2015}%
  \BibitemOpen
  \bibfield  {author} {\bibinfo {author} {\bibfnamefont {E.}~\bibnamefont
  {Diamanti}}\ and\ \bibinfo {author} {\bibfnamefont {A.}~\bibnamefont
  {Leverrier}},\ }\bibfield  {title} {\bibinfo {title} {{Distributing Secret
  Keys with Quantum Continuous Variables: Principle, Security and
  Implementations}},\ }\href {https://doi.org/10.3390/e17096072} {\bibfield
  {journal} {\bibinfo  {journal} {Entropy}\ }\textbf {\bibinfo {volume} {17}},\
  \bibinfo {pages} {6072–6092} (\bibinfo {year} {2015})}\BibitemShut
  {NoStop}%
\bibitem [{\citenamefont {Leverrier}\ \emph {et~al.}(2013)\citenamefont
  {Leverrier}, \citenamefont {Garc\'{\i}a-Patr\'on}, \citenamefont {Renner},\
  and\ \citenamefont {Cerf}}]{Leverrier_2013}%
  \BibitemOpen
  \bibfield  {author} {\bibinfo {author} {\bibfnamefont {A.}~\bibnamefont
  {Leverrier}}, \bibinfo {author} {\bibfnamefont {R.}~\bibnamefont
  {Garc\'{\i}a-Patr\'on}}, \bibinfo {author} {\bibfnamefont {R.}~\bibnamefont
  {Renner}},\ and\ \bibinfo {author} {\bibfnamefont {N.~J.}\ \bibnamefont
  {Cerf}},\ }\bibfield  {title} {\bibinfo {title} {{Security of
  Continuous-Variable Quantum Key Distribution Against General Attacks}},\
  }\href {https://doi.org/10.1103/PhysRevLett.110.030502} {\bibfield  {journal}
  {\bibinfo  {journal} {Phys. Rev. Lett.}\ }\textbf {\bibinfo {volume} {110}},\
  \bibinfo {pages} {030502} (\bibinfo {year} {2013})}\BibitemShut {NoStop}%
\bibitem [{\citenamefont {Leverrier}(2017)}]{Leverrier_2017}%
  \BibitemOpen
  \bibfield  {author} {\bibinfo {author} {\bibfnamefont {A.}~\bibnamefont
  {Leverrier}},\ }\bibfield  {title} {\bibinfo {title} {Security of
  continuous-variable quantum key distribution via a gaussian de finetti
  reduction},\ }\href {https://doi.org/10.1103/PhysRevLett.118.200501}
  {\bibfield  {journal} {\bibinfo  {journal} {Phys. Rev. Lett.}\ }\textbf
  {\bibinfo {volume} {118}},\ \bibinfo {pages} {200501} (\bibinfo {year}
  {2017})}\BibitemShut {NoStop}%
\bibitem [{\citenamefont {Heid}\ and\ \citenamefont
  {L\"utkenhaus}(2006)}]{Heid_2006}%
  \BibitemOpen
  \bibfield  {author} {\bibinfo {author} {\bibfnamefont {M.}~\bibnamefont
  {Heid}}\ and\ \bibinfo {author} {\bibfnamefont {N.}~\bibnamefont
  {L\"utkenhaus}},\ }\bibfield  {title} {\bibinfo {title} {{Efficiency of
  coherent-state quantum cryptography in the presence of loss: Influence of
  realistic error correction}},\ }\href
  {https://doi.org/10.1103/PhysRevA.73.052316} {\bibfield  {journal} {\bibinfo
  {journal} {Phys. Rev. A}\ }\textbf {\bibinfo {volume} {73}},\ \bibinfo
  {pages} {052316} (\bibinfo {year} {2006})}\BibitemShut {NoStop}%
\bibitem [{\citenamefont {Zhao}\ \emph {et~al.}(2009)\citenamefont {Zhao},
  \citenamefont {Heid}, \citenamefont {Rigas},\ and\ \citenamefont
  {Lütkenhaus}}]{Zhao_2009}%
  \BibitemOpen
  \bibfield  {author} {\bibinfo {author} {\bibfnamefont {Y.-B.}\ \bibnamefont
  {Zhao}}, \bibinfo {author} {\bibfnamefont {M.}~\bibnamefont {Heid}}, \bibinfo
  {author} {\bibfnamefont {J.}~\bibnamefont {Rigas}},\ and\ \bibinfo {author}
  {\bibfnamefont {N.}~\bibnamefont {Lütkenhaus}},\ }\bibfield  {title}
  {\bibinfo {title} {{Asymptotic security of binary modulated
  continuous-variable quantum key distribution under collective attacks}},\
  }\href {https://doi.org/10.1103/physreva.79.012307} {\bibfield  {journal}
  {\bibinfo  {journal} {Phys. Rev. A}\ }\textbf {\bibinfo {volume} {79}},\
  \bibinfo {pages} {012307} (\bibinfo {year} {2009})}\BibitemShut {NoStop}%
\bibitem [{\citenamefont {Sych}\ and\ \citenamefont
  {Leuchs}(2010)}]{Sych_2010}%
  \BibitemOpen
  \bibfield  {author} {\bibinfo {author} {\bibfnamefont {D.}~\bibnamefont
  {Sych}}\ and\ \bibinfo {author} {\bibfnamefont {G.}~\bibnamefont {Leuchs}},\
  }\bibfield  {title} {\bibinfo {title} {Coherent state quantum key
  distribution with multi letter phase-shift keying},\ }\href
  {https://doi.org/10.1088/1367-2630/12/5/053019} {\bibfield  {journal}
  {\bibinfo  {journal} {New J. of Phys.}\ }\textbf {\bibinfo {volume} {12}},\
  \bibinfo {pages} {053019} (\bibinfo {year} {2010})}\BibitemShut {NoStop}%
\bibitem [{\citenamefont {Ghorai}\ \emph {et~al.}(2019)\citenamefont {Ghorai},
  \citenamefont {Grangier}, \citenamefont {Diamanti},\ and\ \citenamefont
  {Leverrier}}]{Ghorai_2019}%
  \BibitemOpen
  \bibfield  {author} {\bibinfo {author} {\bibfnamefont {S.}~\bibnamefont
  {Ghorai}}, \bibinfo {author} {\bibfnamefont {P.}~\bibnamefont {Grangier}},
  \bibinfo {author} {\bibfnamefont {E.}~\bibnamefont {Diamanti}},\ and\
  \bibinfo {author} {\bibfnamefont {A.}~\bibnamefont {Leverrier}},\ }\bibfield
  {title} {\bibinfo {title} {{Asymptotic Security of Continuous-Variable
  Quantum Key Distribution with a Discrete Modulation}},\ }\href
  {https://doi.org/10.1103/PhysRevX.9.021059} {\bibfield  {journal} {\bibinfo
  {journal} {Phys. Rev. X}\ }\textbf {\bibinfo {volume} {9}},\ \bibinfo {pages}
  {021059} (\bibinfo {year} {2019})}\BibitemShut {NoStop}%
\bibitem [{\citenamefont {Lin}\ \emph {et~al.}(2019)\citenamefont {Lin},
  \citenamefont {Upadhyaya},\ and\ \citenamefont {Lütkenhaus}}]{Lin_2019}%
  \BibitemOpen
  \bibfield  {author} {\bibinfo {author} {\bibfnamefont {J.}~\bibnamefont
  {Lin}}, \bibinfo {author} {\bibfnamefont {T.}~\bibnamefont {Upadhyaya}},\
  and\ \bibinfo {author} {\bibfnamefont {N.}~\bibnamefont {Lütkenhaus}},\
  }\bibfield  {title} {\bibinfo {title} {{Asymptotic Security Analysis of
  Discrete-Modulated Continuous-Variable Quantum Key Distribution}},\ }\href
  {https://doi.org/10.1103/physrevx.9.041064} {\bibfield  {journal} {\bibinfo
  {journal} {Phys. Rev. X}\ }\textbf {\bibinfo {volume} {9}},\ \bibinfo {pages}
  {041064} (\bibinfo {year} {2019})}\BibitemShut {NoStop}%
\bibitem [{\citenamefont {Lin}\ and\ \citenamefont
  {Lütkenhaus}(2020)}]{Lin_2020}%
  \BibitemOpen
  \bibfield  {author} {\bibinfo {author} {\bibfnamefont {J.}~\bibnamefont
  {Lin}}\ and\ \bibinfo {author} {\bibfnamefont {N.}~\bibnamefont
  {Lütkenhaus}},\ }\bibfield  {title} {\bibinfo {title} {{Trusted Detector
  Noise Analysis for Discrete Modulation Schemes of Continuous-Variable Quantum
  Key Distribution}},\ }\href
  {https://doi.org/10.1103/physrevapplied.14.064030} {\bibfield  {journal}
  {\bibinfo  {journal} {Phys. Rev. Appl.}\ }\textbf {\bibinfo {volume} {14}},\
  \bibinfo {pages} {064030} (\bibinfo {year} {2020})}\BibitemShut {NoStop}%
\bibitem [{\citenamefont {Upadhyaya}\ \emph {et~al.}(2021)\citenamefont
  {Upadhyaya}, \citenamefont {van Himbeeck}, \citenamefont {Lin},\ and\
  \citenamefont {L\"utkenhaus}}]{Upadhyaya_2021}%
  \BibitemOpen
  \bibfield  {author} {\bibinfo {author} {\bibfnamefont {T.}~\bibnamefont
  {Upadhyaya}}, \bibinfo {author} {\bibfnamefont {T.}~\bibnamefont {van
  Himbeeck}}, \bibinfo {author} {\bibfnamefont {J.}~\bibnamefont {Lin}},\ and\
  \bibinfo {author} {\bibfnamefont {N.}~\bibnamefont {L\"utkenhaus}},\
  }\bibfield  {title} {\bibinfo {title} {{Dimension Reduction in Quantum Key
  Distribution for Continuous- and Discrete-Variable Protocols}},\ }\href
  {https://doi.org/10.1103/PRXQuantum.2.020325} {\bibfield  {journal} {\bibinfo
   {journal} {PRX Quantum}\ }\textbf {\bibinfo {volume} {2}},\ \bibinfo {pages}
  {020325} (\bibinfo {year} {2021})}\BibitemShut {NoStop}%
\bibitem [{\citenamefont {Denys}\ \emph {et~al.}(2021)\citenamefont {Denys},
  \citenamefont {Brown},\ and\ \citenamefont {Leverrier}}]{Denys_2021}%
  \BibitemOpen
  \bibfield  {author} {\bibinfo {author} {\bibfnamefont {A.}~\bibnamefont
  {Denys}}, \bibinfo {author} {\bibfnamefont {P.}~\bibnamefont {Brown}},\ and\
  \bibinfo {author} {\bibfnamefont {A.}~\bibnamefont {Leverrier}},\ }\bibfield
  {title} {\bibinfo {title} {{Explicit Asymptotic Secret Key Rate of
  Continuous-Variable Quantum Key Distribution with an Arbitrary Modulation}},\
  }\href {https://doi.org/10.22331/q-2021-09-13-540} {\bibfield  {journal}
  {\bibinfo  {journal} {Quantum}\ }\textbf {\bibinfo {volume} {5}},\ \bibinfo
  {pages} {540} (\bibinfo {year} {2021})}\BibitemShut {NoStop}%
\bibitem [{\citenamefont {Matsuura}\ \emph {et~al.}(2021)\citenamefont
  {Matsuura}, \citenamefont {Maeda}, \citenamefont {Sasaki},\ and\
  \citenamefont {Koashi}}]{Matsuura_2021}%
  \BibitemOpen
  \bibfield  {author} {\bibinfo {author} {\bibfnamefont {T.}~\bibnamefont
  {Matsuura}}, \bibinfo {author} {\bibfnamefont {K.}~\bibnamefont {Maeda}},
  \bibinfo {author} {\bibfnamefont {T.}~\bibnamefont {Sasaki}},\ and\ \bibinfo
  {author} {\bibfnamefont {M.}~\bibnamefont {Koashi}},\ }\bibfield  {title}
  {\bibinfo {title} {{Finite-size security of continuous-variable quantum key
  distribution with digital signal processing}},\ }\href
  {https://doi.org/10.1038/s41467-020-19916-1} {\bibfield  {journal} {\bibinfo
  {journal} {Nat. Commun.}\ }\textbf {\bibinfo {volume} {12}},\ \bibinfo
  {pages} {252} (\bibinfo {year} {2021})}\BibitemShut {NoStop}%
\bibitem [{\citenamefont {Kanitschar}\ and\ \citenamefont
  {Pacher}(2022)}]{Kanitschar_2021}%
  \BibitemOpen
  \bibfield  {author} {\bibinfo {author} {\bibfnamefont {F.}~\bibnamefont
  {Kanitschar}}\ and\ \bibinfo {author} {\bibfnamefont {C.}~\bibnamefont
  {Pacher}},\ }\bibfield  {title} {\bibinfo {title} {Optimizing
  continuous-variable quantum key distribution with phase-shift keying
  modulation and postselection},\ }\href
  {https://doi.org/10.1103/PhysRevApplied.18.034073} {\bibfield  {journal}
  {\bibinfo  {journal} {Phys. Rev. Applied}\ }\textbf {\bibinfo {volume}
  {18}},\ \bibinfo {pages} {034073} (\bibinfo {year} {2022})}\BibitemShut
  {NoStop}%
\bibitem [{\citenamefont {Lupo}\ and\ \citenamefont
  {Ouyang}(2022)}]{Lupo_2022}%
  \BibitemOpen
  \bibfield  {author} {\bibinfo {author} {\bibfnamefont {C.}~\bibnamefont
  {Lupo}}\ and\ \bibinfo {author} {\bibfnamefont {Y.}~\bibnamefont {Ouyang}},\
  }\bibfield  {title} {\bibinfo {title} {{Quantum Key Distribution with
  Nonideal Heterodyne Detection: Composable Security of Discrete-Modulation
  Continuous-Variable Protocols}},\ }\href
  {https://doi.org/10.1103/PRXQuantum.3.010341} {\bibfield  {journal} {\bibinfo
   {journal} {PRX Quantum}\ }\textbf {\bibinfo {volume} {3}},\ \bibinfo {pages}
  {010341} (\bibinfo {year} {2022})}\BibitemShut {NoStop}%
\bibitem [{\citenamefont {Kanitschar}\ \emph {et~al.}(2023)\citenamefont
  {Kanitschar}, \citenamefont {George}, \citenamefont {Lin}, \citenamefont
  {Upadhyaya},\ and\ \citenamefont {L\"utkenhaus}}]{Kanitschar_2023}%
  \BibitemOpen
  \bibfield  {author} {\bibinfo {author} {\bibfnamefont {F.}~\bibnamefont
  {Kanitschar}}, \bibinfo {author} {\bibfnamefont {I.}~\bibnamefont {George}},
  \bibinfo {author} {\bibfnamefont {J.}~\bibnamefont {Lin}}, \bibinfo {author}
  {\bibfnamefont {T.}~\bibnamefont {Upadhyaya}},\ and\ \bibinfo {author}
  {\bibfnamefont {N.}~\bibnamefont {L\"utkenhaus}},\ }\bibfield  {title}
  {\bibinfo {title} {Finite-size security for discrete-modulated
  continuous-variable quantum key distribution protocols},\ }\href
  {https://doi.org/10.1103/PRXQuantum.4.040306} {\bibfield  {journal} {\bibinfo
   {journal} {PRX Quantum}\ }\textbf {\bibinfo {volume} {4}},\ \bibinfo {pages}
  {040306} (\bibinfo {year} {2023})}\BibitemShut {NoStop}%
\bibitem [{\citenamefont {Bäuml}\ \emph {et~al.}(2023)\citenamefont {Bäuml},
  \citenamefont {García}, \citenamefont {Wright}, \citenamefont {Fawzi},\ and\
  \citenamefont {Acín}}]{Baeuml_2023}%
  \BibitemOpen
  \bibfield  {author} {\bibinfo {author} {\bibfnamefont {S.}~\bibnamefont
  {Bäuml}}, \bibinfo {author} {\bibfnamefont {C.~P.}\ \bibnamefont {García}},
  \bibinfo {author} {\bibfnamefont {V.}~\bibnamefont {Wright}}, \bibinfo
  {author} {\bibfnamefont {O.}~\bibnamefont {Fawzi}},\ and\ \bibinfo {author}
  {\bibfnamefont {A.}~\bibnamefont {Acín}},\ }\href@noop {} {\bibinfo {title}
  {Security of discrete-modulated continuous-variable quantum key
  distribution}} (\bibinfo {year} {2023}),\ \Eprint
  {https://arxiv.org/abs/2303.09255} {arXiv:2303.09255 [quant-ph]} \BibitemShut
  {NoStop}%
\bibitem [{\citenamefont {Cabello}(2000)}]{Cabello_2000}%
  \BibitemOpen
  \bibfield  {author} {\bibinfo {author} {\bibfnamefont {A.}~\bibnamefont
  {Cabello}},\ }\href@noop {} {\bibinfo {title} {Multiparty key distribution
  and secret sharing based on entanglement swapping}} (\bibinfo {year}
  {2000}),\ \Eprint {https://arxiv.org/abs/quant-ph/0009025}
  {arXiv:quant-ph/0009025 [quant-ph]} \BibitemShut {NoStop}%
\bibitem [{\citenamefont {Chen}\ and\ \citenamefont {Lo}(2007)}]{Chen_2008}%
  \BibitemOpen
  \bibfield  {author} {\bibinfo {author} {\bibfnamefont {K.}~\bibnamefont
  {Chen}}\ and\ \bibinfo {author} {\bibfnamefont {H.-K.}\ \bibnamefont {Lo}},\
  }\bibfield  {title} {\bibinfo {title} {Multi-partite quantum cryptographic
  protocols with noisy {GHZ} states},\ }\href@noop {} {\bibfield  {journal}
  {\bibinfo  {journal} {Quantum Info. Comput.}\ }\textbf {\bibinfo {volume}
  {7}},\ \bibinfo {pages} {689–715} (\bibinfo {year} {2007})}\BibitemShut
  {NoStop}%
\bibitem [{\citenamefont {Epping}\ \emph {et~al.}(2017)\citenamefont {Epping},
  \citenamefont {Kampermann}, \citenamefont {Macchiavello},\ and\ \citenamefont
  {Bruß}}]{Epping_2017}%
  \BibitemOpen
  \bibfield  {author} {\bibinfo {author} {\bibfnamefont {M.}~\bibnamefont
  {Epping}}, \bibinfo {author} {\bibfnamefont {H.}~\bibnamefont {Kampermann}},
  \bibinfo {author} {\bibfnamefont {C.}~\bibnamefont {Macchiavello}},\ and\
  \bibinfo {author} {\bibfnamefont {D.}~\bibnamefont {Bruß}},\ }\bibfield
  {title} {\bibinfo {title} {Multi-partite entanglement can speed up quantum
  key distribution in networks},\ }\href
  {https://doi.org/10.1088/1367-2630/aa8487} {\bibfield  {journal} {\bibinfo
  {journal} {New Journal of Physics}\ }\textbf {\bibinfo {volume} {19}},\
  \bibinfo {pages} {093012} (\bibinfo {year} {2017})}\BibitemShut {NoStop}%
\bibitem [{\citenamefont {Grasselli}\ \emph {et~al.}(2018)\citenamefont
  {Grasselli}, \citenamefont {Kampermann},\ and\ \citenamefont
  {Bru{\ss}}}]{Grasselli_2018}%
  \BibitemOpen
  \bibfield  {author} {\bibinfo {author} {\bibfnamefont {F.}~\bibnamefont
  {Grasselli}}, \bibinfo {author} {\bibfnamefont {H.}~\bibnamefont
  {Kampermann}},\ and\ \bibinfo {author} {\bibfnamefont {D.}~\bibnamefont
  {Bru{\ss}}},\ }\bibfield  {title} {\bibinfo {title} {Finite-key effects in
  multipartite quantum key distribution protocols},\ }\href
  {https://doi.org/10.1088/1367-2630/aaec34} {\bibfield  {journal} {\bibinfo
  {journal} {New Journal of Physics}\ }\textbf {\bibinfo {volume} {20}},\
  \bibinfo {pages} {113014} (\bibinfo {year} {2018})}\BibitemShut {NoStop}%
\bibitem [{\citenamefont {Murta}\ \emph {et~al.}(2020)\citenamefont {Murta},
  \citenamefont {Grasselli}, \citenamefont {Kampermann},\ and\ \citenamefont
  {Bru{\ss}}}]{Murta_2020}%
  \BibitemOpen
  \bibfield  {author} {\bibinfo {author} {\bibfnamefont {G.}~\bibnamefont
  {Murta}}, \bibinfo {author} {\bibfnamefont {F.}~\bibnamefont {Grasselli}},
  \bibinfo {author} {\bibfnamefont {H.}~\bibnamefont {Kampermann}},\ and\
  \bibinfo {author} {\bibfnamefont {D.}~\bibnamefont {Bru{\ss}}},\ }\bibfield
  {title} {\bibinfo {title} {Quantum conference key agreement: A review},\
  }\bibfield  {journal} {\bibinfo  {journal} {Advanced Quantum Technologies}\
  }\textbf {\bibinfo {volume} {3}},\ \href
  {https://doi.org/10.1002/qute.202000025} {10.1002/qute.202000025} (\bibinfo
  {year} {2020})\BibitemShut {NoStop}%
\bibitem [{\citenamefont {Bian}\ \emph {et~al.}(2023)\citenamefont {Bian},
  \citenamefont {Zhang}, \citenamefont {Zhou}, \citenamefont {Yu},
  \citenamefont {Li},\ and\ \citenamefont {Guo}}]{Bian_2023}%
  \BibitemOpen
  \bibfield  {author} {\bibinfo {author} {\bibfnamefont {Y.}~\bibnamefont
  {Bian}}, \bibinfo {author} {\bibfnamefont {Y.-C.}\ \bibnamefont {Zhang}},
  \bibinfo {author} {\bibfnamefont {C.}~\bibnamefont {Zhou}}, \bibinfo {author}
  {\bibfnamefont {S.}~\bibnamefont {Yu}}, \bibinfo {author} {\bibfnamefont
  {Z.}~\bibnamefont {Li}},\ and\ \bibinfo {author} {\bibfnamefont
  {H.}~\bibnamefont {Guo}},\ }\href@noop {} {\bibinfo {title} {High-rate
  point-to-multipoint quantum key distribution using coherent states}}
  (\bibinfo {year} {2023}),\ \Eprint {https://arxiv.org/abs/2302.02391}
  {arXiv:2302.02391 [quant-ph]} \BibitemShut {NoStop}%
\bibitem [{\citenamefont {Hajomer}\ \emph {et~al.}(2024)\citenamefont
  {Hajomer}, \citenamefont {Derkach}, \citenamefont {Filip}, \citenamefont
  {Andersen}, \citenamefont {C.~Usenko},\ and\ \citenamefont
  {Gehring}}]{Derkach_2024}%
  \BibitemOpen
  \bibfield  {author} {\bibinfo {author} {\bibfnamefont {A.~A.~E.}\
  \bibnamefont {Hajomer}}, \bibinfo {author} {\bibfnamefont {I.}~\bibnamefont
  {Derkach}}, \bibinfo {author} {\bibfnamefont {R.}~\bibnamefont {Filip}},
  \bibinfo {author} {\bibfnamefont {U.~L.}\ \bibnamefont {Andersen}}, \bibinfo
  {author} {\bibfnamefont {V.}~\bibnamefont {C.~Usenko}},\ and\ \bibinfo
  {author} {\bibfnamefont {T.}~\bibnamefont {Gehring}},\ }\bibfield  {title}
  {\bibinfo {title} {Continuous-variable quantum passive optical network},\
  }\bibfield  {journal} {\bibinfo  {journal} {Light: Science \& Applications}\
  }\textbf {\bibinfo {volume} {13}},\ \href
  {https://doi.org/10.1038/s41377-024-01633-9} {10.1038/s41377-024-01633-9}
  (\bibinfo {year} {2024})\BibitemShut {NoStop}%
\bibitem [{\citenamefont {Coles}\ \emph {et~al.}(2016)\citenamefont {Coles},
  \citenamefont {Metodiev},\ and\ \citenamefont {L\"utkenhaus}}]{Coles_2016}%
  \BibitemOpen
  \bibfield  {author} {\bibinfo {author} {\bibfnamefont {P.~J.}\ \bibnamefont
  {Coles}}, \bibinfo {author} {\bibfnamefont {E.~M.}\ \bibnamefont
  {Metodiev}},\ and\ \bibinfo {author} {\bibfnamefont {N.}~\bibnamefont
  {L\"utkenhaus}},\ }\bibfield  {title} {\bibinfo {title} {Numerical approach
  for unstructured quantum key distribution},\ }\href
  {https://doi.org/10.1038/ncomms11712} {\bibfield  {journal} {\bibinfo
  {journal} {Nat. Commun.}\ }\textbf {\bibinfo {volume} {7}},\ \bibinfo {pages}
  {11712} (\bibinfo {year} {2016})}\BibitemShut {NoStop}%
\bibitem [{\citenamefont {Winick}\ \emph {et~al.}(2018)\citenamefont {Winick},
  \citenamefont {Lütkenhaus},\ and\ \citenamefont {Coles}}]{Winick_2018}%
  \BibitemOpen
  \bibfield  {author} {\bibinfo {author} {\bibfnamefont {A.}~\bibnamefont
  {Winick}}, \bibinfo {author} {\bibfnamefont {N.}~\bibnamefont
  {Lütkenhaus}},\ and\ \bibinfo {author} {\bibfnamefont {P.~J.}\ \bibnamefont
  {Coles}},\ }\bibfield  {title} {\bibinfo {title} {{Reliable numerical key
  rates for quantum key distribution}},\ }\href
  {https://doi.org/10.22331/q-2018-07-26-77} {\bibfield  {journal} {\bibinfo
  {journal} {Quantum}\ }\textbf {\bibinfo {volume} {2}},\ \bibinfo {pages} {77}
  (\bibinfo {year} {2018})}\BibitemShut {NoStop}%
\bibitem [{\citenamefont {Devetak}\ and\ \citenamefont
  {Winter}(2005)}]{Devetak_Winter_2006}%
  \BibitemOpen
  \bibfield  {author} {\bibinfo {author} {\bibfnamefont {I.}~\bibnamefont
  {Devetak}}\ and\ \bibinfo {author} {\bibfnamefont {A.}~\bibnamefont
  {Winter}},\ }\bibfield  {title} {\bibinfo {title} {{Distillation of secret
  key and entanglement from quantum states}},\ }\href
  {https://doi.org/10.1098/rspa.2004.1372} {\bibfield  {journal} {\bibinfo
  {journal} {Proc. R. Soc. A}\ }\textbf {\bibinfo {volume} {461}},\ \bibinfo
  {pages} {207} (\bibinfo {year} {2005})}\BibitemShut {NoStop}%
\bibitem [{\citenamefont {Grant}\ and\ \citenamefont {Boyd}(2014)}]{cvx1}%
  \BibitemOpen
  \bibfield  {author} {\bibinfo {author} {\bibfnamefont {M.}~\bibnamefont
  {Grant}}\ and\ \bibinfo {author} {\bibfnamefont {S.}~\bibnamefont {Boyd}},\
  }\href@noop {} {\bibinfo {title} {{CVX: Matlab Software for Disciplined
  Convex Programming, version 2.1}}},\ \bibinfo {howpublished}
  {\url{http://cvxr.com/cvx}} (\bibinfo {year} {2014})\BibitemShut {NoStop}%
\bibitem [{\citenamefont {Grant}\ and\ \citenamefont {Boyd}(2008)}]{cvx2}%
  \BibitemOpen
  \bibfield  {author} {\bibinfo {author} {\bibfnamefont {M.}~\bibnamefont
  {Grant}}\ and\ \bibinfo {author} {\bibfnamefont {S.}~\bibnamefont {Boyd}},\
  }\bibfield  {title} {\bibinfo {title} {Graph implementations for nonsmooth
  convex programs},\ }in\ \href@noop {} {\emph {\bibinfo {booktitle} {Recent
  Advances in Learning and Control}}},\ \bibinfo {series and number} {Lecture
  Notes in Control and Information Sciences},\ \bibinfo {editor} {edited by\
  \bibinfo {editor} {\bibfnamefont {V.}~\bibnamefont {Blondel}}, \bibinfo
  {editor} {\bibfnamefont {S.}~\bibnamefont {Boyd}},\ and\ \bibinfo {editor}
  {\bibfnamefont {H.}~\bibnamefont {Kimura}}}\ (\bibinfo  {publisher}
  {Springer-Verlag Limited, London},\ \bibinfo {year} {2008})\ pp.\ \bibinfo
  {pages} {95--110},\ \bibinfo {note}
  {\url{http://stanford.edu/~boyd/graph_dcp.html}}\BibitemShut {NoStop}%
\bibitem [{\citenamefont {ApS}(2019)}]{mosek}%
  \BibitemOpen
  \bibfield  {author} {\bibinfo {author} {\bibfnamefont {M.}~\bibnamefont
  {ApS}},\ }\href {http://docs.mosek.com/9.0/toolbox/index.html} {\emph
  {\bibinfo {title} {{The MOSEK optimization toolbox for MATLAB manual. Version
  9.0.}}}} (\bibinfo {year} {2019})\BibitemShut {NoStop}%
\bibitem [{\citenamefont {Hajomer}\ \emph {et~al.}(2025)\citenamefont
  {Hajomer}, \citenamefont {Kanitschar}, \citenamefont {Jain}, \citenamefont
  {Hentschel}, \citenamefont {Zhang}, \citenamefont {Lütkenhaus},
  \citenamefont {Andersen}, \citenamefont {Pacher},\ and\ \citenamefont
  {Gehring}}]{Hajomer_Kanitschar_2024}%
  \BibitemOpen
  \bibfield  {author} {\bibinfo {author} {\bibfnamefont {A.~A.~E.}\
  \bibnamefont {Hajomer}}, \bibinfo {author} {\bibfnamefont {F.}~\bibnamefont
  {Kanitschar}}, \bibinfo {author} {\bibfnamefont {N.}~\bibnamefont {Jain}},
  \bibinfo {author} {\bibfnamefont {M.}~\bibnamefont {Hentschel}}, \bibinfo
  {author} {\bibfnamefont {R.}~\bibnamefont {Zhang}}, \bibinfo {author}
  {\bibfnamefont {N.}~\bibnamefont {Lütkenhaus}}, \bibinfo {author}
  {\bibfnamefont {U.~L.}\ \bibnamefont {Andersen}}, \bibinfo {author}
  {\bibfnamefont {C.}~\bibnamefont {Pacher}},\ and\ \bibinfo {author}
  {\bibfnamefont {T.}~\bibnamefont {Gehring}},\ }\bibfield  {title} {\bibinfo
  {title} {Experimental composable key distribution using discrete-modulated
  continuous variable quantum cryptography},\ }\bibfield  {journal} {\bibinfo
  {journal} {Light: Science \& Applications}\ }\textbf {\bibinfo {volume}
  {14}},\ \href {https://doi.org/10.1038/s41377-025-01924-9}
  {10.1038/s41377-025-01924-9} (\bibinfo {year} {2025})\BibitemShut {NoStop}%
\bibitem [{\citenamefont {Martinez-Mateo}\ \emph {et~al.}(2012)\citenamefont
  {Martinez-Mateo}, \citenamefont {Elkouss},\ and\ \citenamefont
  {Martin}}]{Martinez_2013}%
  \BibitemOpen
  \bibfield  {author} {\bibinfo {author} {\bibfnamefont {J.}~\bibnamefont
  {Martinez-Mateo}}, \bibinfo {author} {\bibfnamefont {D.}~\bibnamefont
  {Elkouss}},\ and\ \bibinfo {author} {\bibfnamefont {V.}~\bibnamefont
  {Martin}},\ }\bibfield  {title} {\bibinfo {title} {Blind reconciliation},\
  }\href@noop {} {\bibfield  {journal} {\bibinfo  {journal} {Quantum Info.
  Comput.}\ }\textbf {\bibinfo {volume} {12}},\ \bibinfo {pages} {791–812}
  (\bibinfo {year} {2012})}\BibitemShut {NoStop}%
\bibitem [{\citenamefont {Kržič}\ \emph {et~al.}(2023)\citenamefont
  {Kržič}, \citenamefont {Sharma}, \citenamefont {Spiess}, \citenamefont
  {Chandrashekara}, \citenamefont {Töpfer}, \citenamefont {Sauer},
  \citenamefont {González-Martín~del Campo}, \citenamefont {Kopf},
  \citenamefont {Petscharnig}, \citenamefont {Grafenauer}, \citenamefont
  {Lieger}, \citenamefont {Ömer}, \citenamefont {Pacher}, \citenamefont
  {Berlich}, \citenamefont {Peschel}, \citenamefont {Damm}, \citenamefont
  {Risse}, \citenamefont {Goy}, \citenamefont {Rieländer}, \citenamefont
  {Tünnermann},\ and\ \citenamefont {Steinlechner}}]{Krzic_2023}%
  \BibitemOpen
  \bibfield  {author} {\bibinfo {author} {\bibfnamefont {A.}~\bibnamefont
  {Kržič}}, \bibinfo {author} {\bibfnamefont {S.}~\bibnamefont {Sharma}},
  \bibinfo {author} {\bibfnamefont {C.}~\bibnamefont {Spiess}}, \bibinfo
  {author} {\bibfnamefont {U.}~\bibnamefont {Chandrashekara}}, \bibinfo
  {author} {\bibfnamefont {S.}~\bibnamefont {Töpfer}}, \bibinfo {author}
  {\bibfnamefont {G.}~\bibnamefont {Sauer}}, \bibinfo {author} {\bibfnamefont
  {L.~J.}\ \bibnamefont {González-Martín~del Campo}}, \bibinfo {author}
  {\bibfnamefont {T.}~\bibnamefont {Kopf}}, \bibinfo {author} {\bibfnamefont
  {S.}~\bibnamefont {Petscharnig}}, \bibinfo {author} {\bibfnamefont
  {T.}~\bibnamefont {Grafenauer}}, \bibinfo {author} {\bibfnamefont
  {R.}~\bibnamefont {Lieger}}, \bibinfo {author} {\bibfnamefont
  {B.}~\bibnamefont {Ömer}}, \bibinfo {author} {\bibfnamefont
  {C.}~\bibnamefont {Pacher}}, \bibinfo {author} {\bibfnamefont
  {R.}~\bibnamefont {Berlich}}, \bibinfo {author} {\bibfnamefont
  {T.}~\bibnamefont {Peschel}}, \bibinfo {author} {\bibfnamefont
  {C.}~\bibnamefont {Damm}}, \bibinfo {author} {\bibfnamefont {S.}~\bibnamefont
  {Risse}}, \bibinfo {author} {\bibfnamefont {M.}~\bibnamefont {Goy}}, \bibinfo
  {author} {\bibfnamefont {D.}~\bibnamefont {Rieländer}}, \bibinfo {author}
  {\bibfnamefont {A.}~\bibnamefont {Tünnermann}},\ and\ \bibinfo {author}
  {\bibfnamefont {F.}~\bibnamefont {Steinlechner}},\ }\bibfield  {title}
  {\bibinfo {title} {Towards metropolitan free-space quantum networks},\
  }\bibfield  {journal} {\bibinfo  {journal} {npj Quantum Information}\
  }\textbf {\bibinfo {volume} {9}},\ \href
  {https://doi.org/10.1038/s41534-023-00754-0} {10.1038/s41534-023-00754-0}
  (\bibinfo {year} {2023})\BibitemShut {NoStop}%
\bibitem [{\citenamefont {Johansson}\ \emph {et~al.}(1994)\citenamefont
  {Johansson}, \citenamefont {Kabatianskii},\ and\ \citenamefont
  {Smeets}}]{Johansson94polyhashing}%
  \BibitemOpen
  \bibfield  {author} {\bibinfo {author} {\bibfnamefont {T.}~\bibnamefont
  {Johansson}}, \bibinfo {author} {\bibfnamefont {G.}~\bibnamefont
  {Kabatianskii}},\ and\ \bibinfo {author} {\bibfnamefont {B.}~\bibnamefont
  {Smeets}},\ }\bibfield  {title} {\bibinfo {title} {On the relation between
  a-codes and codes correcting independent errors},\ }in\ \href@noop {} {\emph
  {\bibinfo {booktitle} {Advances in Cryptology --- EUROCRYPT '93}}},\ \bibinfo
  {editor} {edited by\ \bibinfo {editor} {\bibfnamefont {T.}~\bibnamefont
  {Helleseth}}}\ (\bibinfo  {publisher} {Springer Berlin Heidelberg},\ \bibinfo
  {address} {Berlin, Heidelberg},\ \bibinfo {year} {1994})\ pp.\ \bibinfo
  {pages} {1--11}\BibitemShut {NoStop}%
\bibitem [{\citenamefont {Pan}\ \emph {et~al.}(2025)\citenamefont {Pan},
  \citenamefont {Bian}, \citenamefont {Li}, \citenamefont {Xu}, \citenamefont
  {Ma}, \citenamefont {Wang}, \citenamefont {Luo}, \citenamefont {Ye},
  \citenamefont {Pi}, \citenamefont {Yang} \emph {et~al.}}]{Pan_2025}%
  \BibitemOpen
  \bibfield  {author} {\bibinfo {author} {\bibfnamefont {Y.}~\bibnamefont
  {Pan}}, \bibinfo {author} {\bibfnamefont {Y.}~\bibnamefont {Bian}}, \bibinfo
  {author} {\bibfnamefont {Y.}~\bibnamefont {Li}}, \bibinfo {author}
  {\bibfnamefont {X.}~\bibnamefont {Xu}}, \bibinfo {author} {\bibfnamefont
  {L.}~\bibnamefont {Ma}}, \bibinfo {author} {\bibfnamefont {H.}~\bibnamefont
  {Wang}}, \bibinfo {author} {\bibfnamefont {Y.}~\bibnamefont {Luo}}, \bibinfo
  {author} {\bibfnamefont {T.}~\bibnamefont {Ye}}, \bibinfo {author}
  {\bibfnamefont {Y.}~\bibnamefont {Pi}}, \bibinfo {author} {\bibfnamefont
  {J.}~\bibnamefont {Yang}}, \emph {et~al.},\ }\bibfield  {title} {\bibinfo
  {title} {High-rate 16-node quantum access network based on a passive optical
  network},\ }\href@noop {} {\bibfield  {journal} {\bibinfo  {journal}
  {Optica}\ }\textbf {\bibinfo {volume} {12}},\ \bibinfo {pages} {953}
  (\bibinfo {year} {2025})}\BibitemShut {NoStop}%
\bibitem [{\citenamefont {Kanitschar}(2021)}]{Kanitschar_Thesis_2021}%
  \BibitemOpen
  \bibfield  {author} {\bibinfo {author} {\bibfnamefont {F.~P.}\ \bibnamefont
  {Kanitschar}},\ }\emph {\bibinfo {title} {{Postselection Strategies for
  CV-QKD protocols with Phase-Shift Keying Modulation}}},\ \href
  {https://repositum.tuwien.at/handle/20.500.12708/18394} {Master's thesis},\
  \bibinfo  {school} {TU Wien} (\bibinfo {year} {2021})\BibitemShut {NoStop}%
\bibitem [{\citenamefont {Curty}\ \emph {et~al.}(2004)\citenamefont {Curty},
  \citenamefont {Lewenstein},\ and\ \citenamefont {L\"utkenhaus}}]{Curty_2004}%
  \BibitemOpen
  \bibfield  {author} {\bibinfo {author} {\bibfnamefont {M.}~\bibnamefont
  {Curty}}, \bibinfo {author} {\bibfnamefont {M.}~\bibnamefont {Lewenstein}},\
  and\ \bibinfo {author} {\bibfnamefont {N.}~\bibnamefont {L\"utkenhaus}},\
  }\bibfield  {title} {\bibinfo {title} {{Entanglement as a Precondition for
  Secure Quantum Key Distribution}},\ }\href
  {https://doi.org/10.1103/PhysRevLett.92.217903} {\bibfield  {journal}
  {\bibinfo  {journal} {Phys. Rev. Lett.}\ }\textbf {\bibinfo {volume} {92}},\
  \bibinfo {pages} {217903} (\bibinfo {year} {2004})}\BibitemShut {NoStop}%
\bibitem [{\citenamefont {Ferenczi}\ and\ \citenamefont
  {L\"utkenhaus}(2012)}]{Ferenzci_2012}%
  \BibitemOpen
  \bibfield  {author} {\bibinfo {author} {\bibfnamefont {A.}~\bibnamefont
  {Ferenczi}}\ and\ \bibinfo {author} {\bibfnamefont {N.}~\bibnamefont
  {L\"utkenhaus}},\ }\bibfield  {title} {\bibinfo {title} {{Symmetries in
  quantum key distribution and the connection between optimal attacks and
  optimal cloning}},\ }\href {https://doi.org/10.1103/PhysRevA.85.052310}
  {\bibfield  {journal} {\bibinfo  {journal} {Phys. Rev. A}\ }\textbf {\bibinfo
  {volume} {85}},\ \bibinfo {pages} {052310} (\bibinfo {year}
  {2012})}\BibitemShut {NoStop}%
\bibitem [{\citenamefont {Upadhyaya}(2021)}]{Upadhyaya_Thesis_2021}%
  \BibitemOpen
  \bibfield  {author} {\bibinfo {author} {\bibfnamefont {T.}~\bibnamefont
  {Upadhyaya}},\ }\emph {\bibinfo {title} {{Tools for the Security Analysis of
  Quantum Key Distribution in Infinite Dimensions}}},\ \href
  {http://hdl.handle.net/10012/17209} {Master's thesis} (\bibinfo {year}
  {2021})\BibitemShut {NoStop}%
\bibitem [{\citenamefont {Frank}\ and\ \citenamefont
  {Wolfe}(1956)}]{Frank_Wolfe_1956}%
  \BibitemOpen
  \bibfield  {author} {\bibinfo {author} {\bibfnamefont {M.}~\bibnamefont
  {Frank}}\ and\ \bibinfo {author} {\bibfnamefont {P.}~\bibnamefont {Wolfe}},\
  }\bibfield  {title} {\bibinfo {title} {An algorithm for quadratic
  programming},\ }\href
  {https://EconPapers.repec.org/RePEc:wly:navlog:v:3:y:1956:i:1-2:p:95-110}
  {\bibfield  {journal} {\bibinfo  {journal} {Nav. Res. Logist. Q.}\ }\textbf
  {\bibinfo {volume} {3}},\ \bibinfo {pages} {95} (\bibinfo {year}
  {1956})}\BibitemShut {NoStop}%
\bibitem [{\citenamefont {Coles}(2012)}]{Coles_2012}%
  \BibitemOpen
  \bibfield  {author} {\bibinfo {author} {\bibfnamefont {P.~J.}\ \bibnamefont
  {Coles}},\ }\bibfield  {title} {\bibinfo {title} {Unification of different
  views of decoherence and discord},\ }\bibfield  {journal} {\bibinfo
  {journal} {Physical Review A}\ }\textbf {\bibinfo {volume} {85}},\ \href
  {https://doi.org/10.1103/physreva.85.042103} {10.1103/physreva.85.042103}
  (\bibinfo {year} {2012})\BibitemShut {NoStop}%
\bibitem [{\citenamefont {Kanitschar}(2022)}]{Kanitschar_Thesis_2022}%
  \BibitemOpen
  \bibfield  {author} {\bibinfo {author} {\bibfnamefont {F.~P.}\ \bibnamefont
  {Kanitschar}},\ }\emph {\bibinfo {title} {{Finite-size security proof for
  discrete-modulated CV-QKD protocols}}},\ \href
  {https://repositum.tuwien.at/handle/20.500.12708/136293} {Master's thesis},\
  \bibinfo  {school} {TU Wien} (\bibinfo {year} {2022})\BibitemShut {NoStop}%
\bibitem [{\citenamefont {Leverrier}(2023)}]{Leverrier_2023}%
  \BibitemOpen
  \bibfield  {author} {\bibinfo {author} {\bibfnamefont {A.}~\bibnamefont
  {Leverrier}},\ }\href@noop {} {\bibinfo {title} {Information reconciliation
  for discretely-modulated continuous-variable quantum key distribution}}
  (\bibinfo {year} {2023}),\ \Eprint {https://arxiv.org/abs/2310.17548}
  {arXiv:2310.17548 [quant-ph]} \BibitemShut {NoStop}%
\bibitem [{\citenamefont {Kanitschar}\ \emph {et~al.}(2024)\citenamefont
  {Kanitschar}, \citenamefont {Bergmayr-Mann}, \citenamefont {Pivoluska},\ and\
  \citenamefont {Huber}}]{Kanitschar_2023_HD}%
  \BibitemOpen
  \bibfield  {author} {\bibinfo {author} {\bibfnamefont {F.}~\bibnamefont
  {Kanitschar}}, \bibinfo {author} {\bibfnamefont {A.}~\bibnamefont
  {Bergmayr-Mann}}, \bibinfo {author} {\bibfnamefont {M.}~\bibnamefont
  {Pivoluska}},\ and\ \bibinfo {author} {\bibfnamefont {M.}~\bibnamefont
  {Huber}},\ }\bibfield  {title} {\bibinfo {title} {Harnessing high-dimensional
  temporal entanglement using limited interferometric setups},\ }\bibfield
  {journal} {\bibinfo  {journal} {Physical Review Applied}\ }\textbf {\bibinfo
  {volume} {22}},\ \href {https://doi.org/10.1103/physrevapplied.22.054054}
  {10.1103/physrevapplied.22.054054} (\bibinfo {year} {2024})\BibitemShut
  {NoStop}%
\end{thebibliography}%
	
\end{document}